\documentclass[10pt,aps,prb,onecolumn,floatfix,nofootinbib,superscriptaddress,notitlepage]{revtex4-2}

\usepackage{amsmath,amsfonts}
\usepackage{algorithmic}
\usepackage{graphicx}
\usepackage{subcaption}
\usepackage{textcomp}
\usepackage[dvipsnames]{xcolor}
\usepackage{dsfont} 
\usepackage{xspace} 
\usepackage[colorlinks=true, linkcolor=blue, citecolor=blue, urlcolor=blue]{hyperref}
\usepackage[capitalise, nameinlink]{cleveref}
\usepackage[most]{tcolorbox}
\usepackage{array}
\usepackage{adjustbox}
\usepackage[normalem]{ulem}

\usepackage{amsthm}
\newtheorem{theorem}{Theorem}
\newtheorem{definition}{Definition}
\newtheorem{corollary}{Corollary}
\newtheorem{lemma}{Lemma}

\newcolumntype{C}[1]{>{\centering\arraybackslash}p{#1}}

\usepackage{tikz}
\usepackage{venndiagram}
\usetikzlibrary{decorations.pathreplacing, calligraphy, decorations.markings}
\usetikzlibrary{arrows, automata, positioning}
\usetikzlibrary{calc}

\definecolor{purple}{HTML}{E5E3F5}
\definecolor{amaranth}{HTML}{F5CFD8}
\definecolor{yellow}{HTML}{FFEAB6}

\newcommand{\MPSTensor}[3]{
	\begin{scope}[shift={(#1)}]
		\draw (-1,0) -- (1,0);
		\draw (0,1) -- (0,0);
		\filldraw[fill=#3] (-1/2,-1/2) -- (-1/2,1/2) -- (1/2,1/2) -- (1/2,-1/2) -- (-1/2,-1/2);
		\draw (0,0) node {\scriptsize #2};
	\end{scope}
}

\newcommand{\MPSTensorRect}[4]{
	\begin{scope}[shift={(#1)}]
		\draw (-1-#4,0) -- (1+#4,0);
		\draw (0,1) -- (0,0);
		\filldraw[fill=#3] (-1/2-#4,-1/2) -- (-1/2-#4,1/2) -- (1/2+#4,1/2) -- (1/2+#4,-1/2) -- (-1/2-#4,-1/2);
		\draw (0,0) node {\scriptsize #2};
	\end{scope}
}

\newcommand{\LeftMPSTensor}[3]{
	\begin{scope}[shift={(#1)}]
		\draw (0,0) -- (1,0);
		\draw (0,1) -- (0,0);
		\filldraw[fill=#3] (-1/2,-1/2) -- (-1/2,1/2) -- (1/2,1/2) -- (1/2,-1/2) -- (-1/2,-1/2);
		\draw (0,0) node {\scriptsize #2};
	\end{scope}
}

\newcommand{\LeftMPSTensorRect}[4]{
	\begin{scope}[shift={(#1)}]
		\draw (0,0) -- (1+#4,0);
		\draw (0,1) -- (0,0);
		\filldraw[fill=#3] (-1/2-#4,-1/2) -- (-1/2-#4,1/2) -- (1/2+#4,1/2) -- (1/2+#4,-1/2) -- (-1/2-#4,-1/2);
		\draw (0,0) node {\scriptsize #2};
	\end{scope}
}

\newcommand{\RightMPSTensor}[3]{
	\begin{scope}[shift={(#1)}]
		\draw (-1,0) -- (0,0);
		\draw (0,1) -- (0,0);
		\filldraw[fill=#3] (-1/2,-1/2) -- (-1/2,1/2) -- (1/2,1/2) -- (1/2,-1/2) -- (-1/2,-1/2);
		\draw (0,0) node {\scriptsize #2};
	\end{scope}
}

\newcommand{\RightMPSTensorRect}[4]{
	\begin{scope}[shift={(#1)}]
		\draw (-1-#4,0) -- (0,0);
		\draw (0,1) -- (0,0);
		\filldraw[fill=#3] (-1/2-#4,-1/2) -- (-1/2-#4,1/2) -- (1/2+#4,1/2) -- (1/2+#4,-1/2) -- (-1/2-#4,-1/2);
		\draw (0,0) node {\scriptsize #2};
	\end{scope}
}

\newcommand{\FullMPS}[3]{
	\begin{scope}[shift={(#1)}]
		\draw[shift={(0,0)},dotted] (0,0) -- (4.5,0);
        \MPSTensor{0,0}{#2}{#3}
        \MPSTensor{1.5,0}{#2}{#3}
        \MPSTensor{5,0}{#2}{#3}
	\end{scope}
}
\def\BibTeX{{\rm B\kern-.05em{\sc i\kern-.025em b}\kern-.08em
    T\kern-.1667em\lower.7ex\hbox{E}\kern-.125emX}}

\definecolor{keynoteblue}{RGB}{40,73,93} 
\definecolor{Arm}{rgb}{0.3,0.2,0.7}

\definecolor{carnelian}{rgb}{0.7, 0.11, 0.11}
\definecolor{darkspringgreen}{rgb}{0.09, 0.45, 0.27}

\newcommand{\quoted}[1]{``#1''}
\usepackage{xspace}
\newcommand{\SeqRLSP}{SeqRLSP\xspace}
\newcommand{\TreeRLSP}{TreeRLSP\xspace}
\newcommand{\spar}{s}
\newcommand{\CNOT}{\mathrm{CNOT}}
\newcommand{\ceil}[1]{\lceil #1 \rceil}

\newcommand{\ket}[1]{| #1 \rangle}
\newcommand{\bra}[1]{\langle #1 |}

\AtBeginDocument{%
  \providecommand\BibTeX{{%
    Bib\TeX}}}

\begin{document}

\title{Compiling Quantum Regular Language States}

\author{Armando Bellante}
\email{armando.bellante@mpq.mpg.de}

\author{Reinis Irmejs}
\email{reinis.irmejs@mpq.mpg.de}

\author{Marta Florido-Llinàs}
\author{María Cea Fernández}
\author{Marianna Crupi}
\affiliation{Max-Planck-Institut für Quantenoptik, Hans-Kopfermann-Straße 1, D-85748 Garching, Germany}
\affiliation{Munich Center for Quantum Science and Technology (MCQST), Schellingstraße 4, D-80799 Munich, Germany}

\author{Matthew Kiser}
\affiliation{TUM School of Natural Sciences, Technical University of Munich, Boltzmannstr. 10, 85748 Garching, Germany}
\affiliation{IQM Quantum Computers, Georg-Brauchle-Ring 23-25, 80992 Munich, Germany}

\author{J. Ignacio Cirac}
\affiliation{Max-Planck-Institut für Quantenoptik, Hans-Kopfermann-Straße 1, D-85748 Garching, Germany}
\affiliation{Munich Center for Quantum Science and Technology (MCQST), Schellingstraße 4, D-80799 Munich, Germany}

\begin{abstract}
    State preparation compilers for quantum computers typically sit at two extremes: general-purpose routines that treat the target as an opaque amplitude vector, and bespoke constructions for a handful of well-known state families. 
    We ask whether a compiler can instead accept simple, structure-aware specifications while providing predictable resource guarantees. 
    We answer this by designing and implementing a quantum state-preparation compiler for regular language states (RLS): uniform superpositions over bitstrings accepted by a regular description, and their complements. Users describe the target state via (i) a finite set of bitstrings, (ii) a regular expression, or (iii) a deterministic finite automaton (DFA), optionally with a complement flag. 
    By translating the input to a DFA, minimizing it, and mapping it to an optimal matrix product state (MPS), the compiler obtains an intermediate representation (IR) that exposes and compresses hidden structure.
    The efficient DFA representation and minimization offloads expensive linear algebra computation in exchange of simpler automata manipulations.
    The combination of the regular-language frontend and this IR gives concise specifications not only for RLS but also for their complements that might otherwise require exponentially large state descriptions. 
    This enables state preparation of an RLS or its complement with the same asymptotic resources and compile time, which to our knowledge is not supported by existing compilers. 
    We outline two hardware-aware backends: SeqRLSP, which yields linear-depth, ancilla-free circuits for linear nearest-neighbor architectures via sequential generation, and TreeRLSP, which achieves logarithmic depth on all-to-all connectivity via a tree tensor network. On the theory side, we prove circuit-depth and gate-count bounds that scale with the system size and the maximal Schmidt rank of the target state, and we give compile-time bounds that expose the benefit of the initial DFA representation.
    We implement the full pipeline and evaluate it on Dicke and W states, random uniform superpositions, and complement states, comparing against general-purpose, sparse-state, and specialized baselines.
\end{abstract}

\maketitle

\section{Introduction}
A quantum computer is a programmable physical system that not only obeys the laws of quantum mechanics but \emph{manifests} them at device scale, making superposition and entanglement controllable and exploitable.
This capability expands the way we can process information, promising advantages over conventional classical computers. In particular, there exist proven exponential advantages over the best-known classical algorithms on several problems, with applications ranging from cryptanalysis via integer factorization and the discrete logarithm problem~\cite{shor1994algorithms}, to problems in algebraic number theory~\cite{hallgren2007polynomial}, and simulation of many-body physics~\cite{lloydsimulation, somma2025shadow}. 
Under appropriate input models, polynomial and even super-polynomial advantages also arise in optimization~\cite{harrow2009quantum, chakraborty2023quantum, chen2023quantum, bellante2022quantum, Abbas_2024, jordan2025optimization}, statistics~\cite{montanaro2015quantum, cornelissen2022near}, and machine learning~\cite{rebentrost2014quantum, kerenidis2016quantum, Biamonte_2017, du2025, lewis2025}.

Over the past years, quantum computers have steadily improved in scale and fidelity~\cite{bluvstein2024logical, belowthreshold}, increasing the need for automatic methods that efficiently map high-level user intent to low-level quantum circuits, both in terms of \emph{compile-time} and in \emph{quantum resources} (\emph{circuit depth}, \emph{gate count}, \emph{qubit count}).
In practice, programming these devices resembles a Field-Programmable Gate Array (FPGA) synthesis workflow.
In the quantum circuit model, an algorithm is a tight classical–quantum loop: a classical host compiles a quantum circuit, the quantum processor initializes the qubits in the all-zero state, applies the prescribed quantum gates and measures the final state, returning the sampled bitstring. The host aggregates outcomes, interprets the information, and possibly adapts the next circuit until the task is solved.

Many algorithms are most naturally phrased as having query access to a circuit that prepares a target state from the all-zero state: linear system solvers assume preparation of the right-hand side vector~\cite{harrow2009quantum}; amplitude amplification and estimation assume an oracle that contains and flags the \quoted{good} subspace~\cite{brassard2000quantum}; data-driven methods assume access to a loading circuit that creates a superposition encoding the dataset~\cite{kerenidis2019q, bellante2022quantumqadra, kerenidis2020classification}.

Since state preparation is such an important part of many quantum algorithms, much like {the function prologue that a classical compiler injects around every call, it is crucial that this subroutine is compiled efficiently. 
Particularly, the host needs to synthesize a quantum circuit that prepares the target state, starting from some user description.
In general, an exact classical description of a quantum state on $N$ quantum bits (qubits) requires a specification of up to $2^N$ complex amplitudes, due to the nature of quantum superposition. This poses a fundamental barrier as the number of qubits increase.

In practice, target specifications reach the compiler in two modes. 
In the first, the user supplies an explicit support - a list or dictionary of amplitudes - which is limited to \emph{sparse} quantum states, where the number of non-zero amplitudes is $\ll2^N$, or to \emph{small system sizes}. 
In this case, general-purpose routines treat the input as an arbitrary vector and typically miss exploitable regularities beyond sparsity, yielding inefficient circuits. 
In the second, the user recognizes that the state belongs to a well-studied family (\emph{e.g.}, GHZ~\cite{Greenberger_1989}, W~\cite{Dur_2000_Wstate}, Dicke~\cite{Dicke_1954}) and invokes a specialized routine that synthesizes an efficient purpose-built circuit.

This gap between universal but oblivious methods and bespoke recipes leads to our question:
\begin{center}
    \emph{Can we offer a simple, compact, structure-aware specification that does not require naming a specific state family, yet compiles to circuits with predictable and efficient resource bounds?}
\end{center}

We answer in the middle ground and focus on a broad but disciplined class of targets, Regular Language States (RLS): the \emph{uniform} superpositions over length-$N$ bitstrings recognized by a regular description, and, symmetrically, their \emph{complements} - the uniform superposition of all length-$N$ bitstrings except those described.
In this work, we limit our scope to uniform superpositions rather than arbitrary complex coefficients. 
This choice keeps the class easier to analyze and still covers many representative families (GHZ, W, Dicke, fixed parity/weight, simple prefix/suffix) \cite{Dur_2000_Wstate, Dicke_1954, Greenberger_1989}.

Prior work connects regular languages (RLs) to matrix product states (MPS)~\cite{florido2024regular, crosswhite_2008}, a special class of tensor networks (TNs)~\cite{orus_2014_introTNs}, and, separately, shows how to realize a given MPS as a circuit via sequential generation~\cite{schon_2005_sequentialMPS}, tree/parallel structures~\cite{Cirac_2009_trees1, malz2024preparation}, or measurement-and-feedback schemes~\cite{piroli2024approximating}.
We bridge the gap between regular descriptions and state synthesis, combining these independent facts into an end-to-end \emph{compiler pipeline}~(\Cref{fig:pipeline}). 
Users can provide a regular description for strings of length $N$ as 
\begin{itemize}
    \item[---] a set (or list, or dictionary) of bitstrings,
    \item[---] a regular expression (regex),
    \item[---] a Deterministic Finite Automaton (DFA) accepting the language.
\end{itemize}
Optionally, they may request the complement of any description.
As shown in \cref{fig:pipeline}, the compiler canonicalizes the input to a minimal DFA, translates the DFA to an MPS, constructs quantum operations from the MPS, and translates them into a quantum circuit.
Conceptually, the frontend maps user specifications into a DFA/MPS Intermediate Representation (IR), and the backend chooses an MPS-to-circuit realization based on hardware constraints.
This yields a modular pipeline that leverages results from formal-language theory, TNs, and MPS synthesis to automatically uncover structure and optimize the circuit synthesis.
On the frontend, regular descriptions offer a generic yet structure-aware specification, while the IR and the backend ensure that the compilation provides predictable and efficient resource bounds.

\begin{figure}[t]
    \centering
    \includegraphics[width=\linewidth]{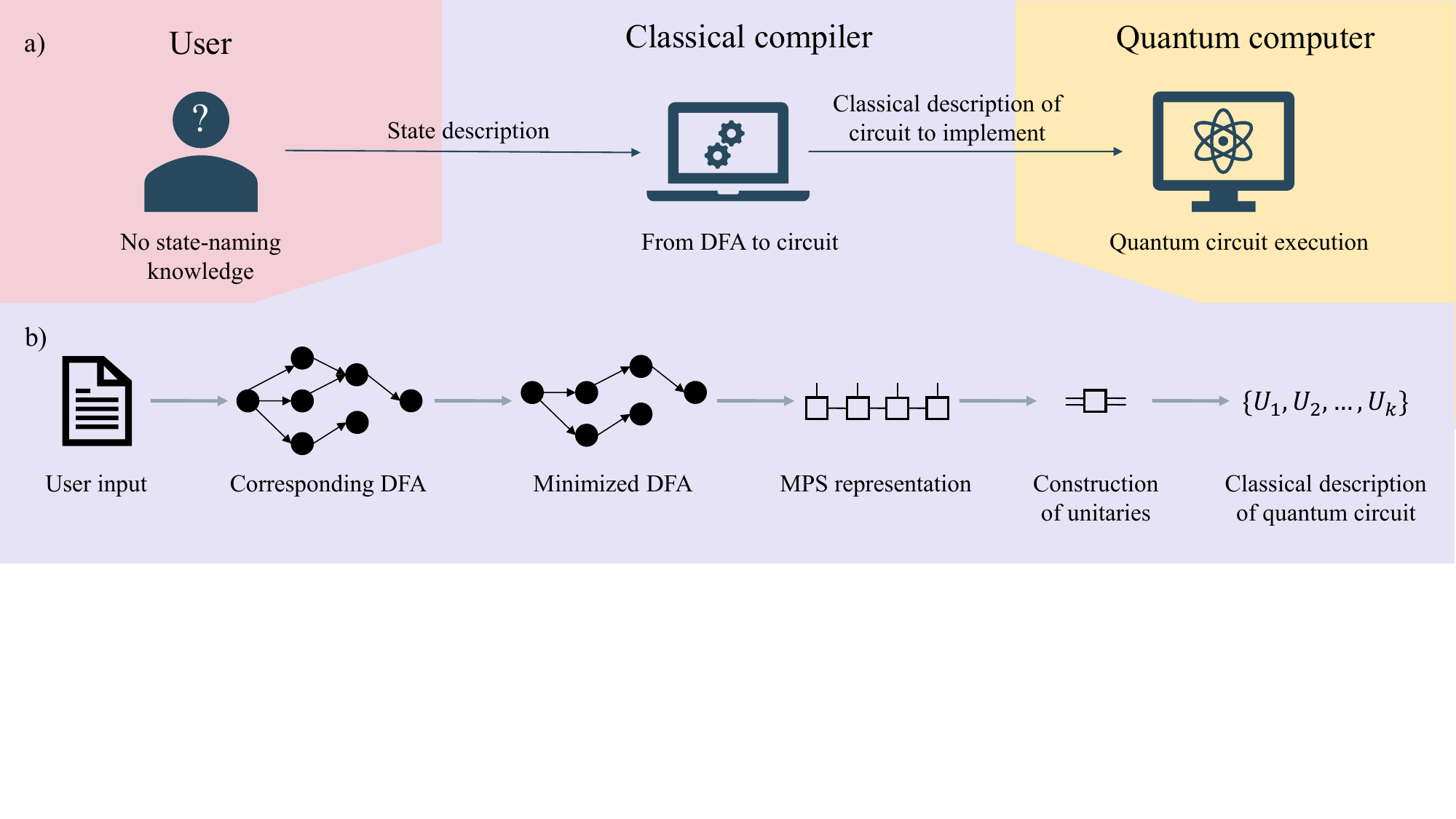} \vspace{-9em}
    \caption{(a) A schematic illustration of the general setting. It starts with a user, possibly without prior knowledge of the state’s name or structure, providing a description of the desired state—such as a set of strings, a regex, or a DFA, and (optionally) a complement flag—to the classical compiler. The compiler then generates a quantum circuit using a gate set compatible with the target hardware, which executes the resulting instructions. (b) Summary of our compilation pipeline using DFAs and MPSs. The user input (regex, DFA or set of strings) is converted into a DFA, which is then minimized and translated into an MPS. The MPS is decomposed into unitaries and compiled into a gate instruction set to be executed on the quantum computer. }
    \label{fig:pipeline}
\end{figure}

Each way of specifying the target state has its own advantages.
First, the user can provide unstructured input as a set, list, or dictionary, and the pipeline will still not treat it as an opaque vector: the frontend converts it into a DFA, placing it in the same IR as the other inputs and enabling the compiler to expose structure.
Second, when users know the structure they want, they can express it directly with a regex or DFA without naming a specific state family: for example, a uniform superposition over all one-hot strings of length $N$ can be written as a simple regex pattern, without knowing the term \quoted{W state}. 
Third, although in this work we focus on RLs, the DFA input also extends use beyond \emph{strictly regular} families: for instance, Motzkin states arise from the (non-regular) Dyck language of well-parenthesized strings~\cite{movassagh2016supercritical, gopalakrishnan2025pushdown}, yet for any fixed length $N$ they can still be concisely expressed by a DFA.
Finally, complements are a native construct: instead of listing all allowed strings, the user can describe a small forbidden language and request the uniform superposition over all remaining bitstrings. 
This keeps the classical description compact even when the quantum support is exponentially large, and, as we will show, does not substantially increase the structural complexity of the state.

The optimization passes on the IR first minimize the DFA and then refine the internal dimensions of the MPS tensors, uncovering structure in the target state.
Minimizing the DFA before acting on the MPS shifts work away from expensive linear-algebra operations on large tensors and towards simpler automata manipulations. 
This makes compilation more efficient.
More broadly, the DFA-MPS correspondence lets us systematically derive compact MPS representations from user-level specifications, rather than assuming that an expert has already provided a good MPS, bridging the gap to prior MPS-based synthesis work.

From the IR, the backend synthesizes hardware-aware circuits. 
In this work we outline two implementations: a sequential backend, \SeqRLSP, for linear nearest-neighbor devices, which achieves linear depth without ancillae~\cite{schon_2005_sequentialMPS}, and a tree-based backend, \TreeRLSP, for all-to-all devices, which achieves logarithmic depth by parallelizing the MPS realization~\cite{Cirac_2009_trees1, malz2024preparation}.

\paragraph{Contributions} In summary, our contributions are:
\begin{enumerate}
    \item \emph{Specification.} We introduce a \emph{regular} description for state preparation (regex/DFA/finite set) together with a \emph{complement} option. This lets users describe large, structured supports concisely without naming families, and extends naturally to fixed-length slices of non-regular families.
    \item \emph{Compilation.} We bridge the gap between the \emph{RLs $\rightarrow$ MPS} and \emph{MPS $\rightarrow$ circuits} literature, providing an end-to-end compiler with a DFA/MPS IR and two hardware-aware backends, \SeqRLSP and \TreeRLSP.
    \item \emph{Resource use and compile time.} We give explicit circuit resource bounds in terms of $N$ and the target state's maximal Schmidt rank $\chi$ (\cref{def:schmidt-rank-bipartition}, \cref{thm:seqRLSP,thm:treeRLSP}), and provide worst-case compile time bounds (\cref{thm: compile-time dfa}, \cref{coro: compile-time set of strings}).
    \item \emph{Theory for complements.} We prove that complements compile with the \emph{same asymptotic quantum cost} as their base description, since their maximal Schmidt rank $\chi$ increases by at most $1$ (\cref{thm:complementRLS}), tying the entanglement structure of a state to that of its complement.
    \item \emph{Evaluation.} We implement the full stack and support our claims by benchmarking against general-purpose, sparse-state, and Dicke/W-specialized baselines across Dicke/W, random uniform superpositions, and complements, reporting \emph{circuit depth}, \emph{gates count}, \emph{ancillae}, and \emph{compile time}.
\end{enumerate}

\paragraph{Document structure} \Cref{sec: background} presents the relevant background, reviewing elements of quantum computation and circuit design, related compilers, RLS, and MPS. We also review existing quantum state compilers.
\Cref{sec:methodology} presents our methodology and outlines each of the steps in our compiler. 
In \Cref{sec: costs and guarantees}, we discuss compile-time complexity and bound the quantum resources for RLS preparation using our compiler, further relating the entanglement structure of an RLS to its complement.
Finally, \Cref{sec:experiments} presents numerical experiments benchmarking our compiler against existing methods, and \cref{sec:conclusions} concludes the paper, highlighting potential future work.

\section{Background}
\label{sec: background}
In this section, we review the essential background for our work, including quantum computation, RLS, finite-state automata, and MPS, as well as related compiler and state-preparation techniques. For deeper treatments, see \citet{nielsen2010quantum} for quantum computing, \citet{Cirac_review_2021} and \citet{orus_2014_introTNs} for MPS and TNs, and \citet{Hopcroft_2001_book_automata} for finite-state automata and RLs.

\subsection{Quantum Computation}

\paragraph{Quantum states} 
A single \emph{qudit} has \emph{physical dimension} $d$ with the standard (computational) basis vectors $\{\ket{0},\dots,\ket{d-1}\}$. These labels are an alphabet $\Sigma$ with $|\Sigma|=d$ (for qubits, $d{=}2$ and $\Sigma=\{0,1\}$).
A register of $N$ qudits is in the Hilbert space $(\mathbb{C}^d)^{\otimes N}$ whose basis vectors are tensor products of the individual qudit states $\ket{x_1}\!\otimes\!\cdots\!\otimes\!\ket{x_N}$ indexed by strings $x=x_1\cdots x_N\in\Sigma^N$; we write $\ket{x}$ for short.
A (pure) quantum state is a unit vector $\ket{\psi}=\sum_{x\in\Sigma^N}\alpha_x\,\ket{x}$ such that $\sum_x |\alpha_x|^2=1,$ for the complex \emph{amplitudes} $\alpha_x$, and a global phase is physically irrelevant. 
Thus, a generic $N$-qudit state is specified by $d^N$ amplitudes. 
For a finite set $S \subseteq\Sigma^N$, the \emph{uniform superposition} is $\ket{u_S}=\frac{1}{\sqrt{|S|}}\sum_{x\in S}\ket{x}$; i.e., $\alpha_x$ is constant on $S$ and $0$ outside.
In particular, for qubits, it is common to write the state $\ket{+}^{\otimes N}=\frac{1}{\sqrt{2^N}}\sum_{x\in\{0,1\}^N}\ket{x}$, which is the uniform superposition over all $N$-bit strings.

\paragraph{Unitary evolutions and quantum circuits} 
Closed quantum systems evolve under unitary operators, that is, invertible linear maps on $(\mathbb{C}^d)^{\otimes N}$ that preserve the inner product. 
For state preparation, a unitary $U$ acting on $N$ qudits maps the all-zero state to the target, $U\ket{0^N}=\ket{\psi}$.
The circuit model leverages the fact that any $N$-qubit unitary can be exactly decomposed into \emph{smaller} unitaries acting on $1$ or $2$ qudits, which are called \emph{gates}.
Since hardware allows for only a fixed, finite gate set that must be universal, these arbitrary gates need to be efficiently approximated to a chosen precision during synthesis (\emph{e.g.}, via Solovay-Kitaev~\cite{kitaev1997quantum,dawson2006solovay}). 

\paragraph{Quantum hardware} 
For our purposes, a device is described by (i) a calibrated \emph{native, finite, and universal gate set}, and (ii) a \emph{connectivity graph} indicating which qudits can interact directly. 
Two common connectivity models are \emph{linear nearest-neighbor} LNN, where interactions are only between adjacent sites (necessitating SWAP routing for distant gates), and \emph{all-to-all}, where any pair can interact directly.
These constraints shape circuit cost: in the setting of LNN, routing increases the gate count and depth, but gates can typically be performed in parallel, while all-to-all favors shallower depths, but might not allow parallel gate operations.
The compiler \emph{transpiles} abstract circuits to hardware by decomposing arbitrary gates into native gates with connectivity-aware routing.

\paragraph{Resource count} The cost model mirrors the one of classical boolean circuits: the figures of merit are the \emph{critical circuit depth}, the \emph{total elementary gates count}, and the \emph{overall number of qudits}.
For state preparation, $N$ is a trivial lower bound on the number of qudits; therefore, we quantify the number of auxiliary qudits (also referred to as \emph{ancillae}).

\subsection{State Preparation and Compilers}
\label{sec:literature_review}
In the following, we review some representative methods for quantum state preparation, ranging from general-purpose compilers to algorithms targeting specific structured families of states. We focus on qubit methods, setting $d=2$.

\paragraph{General purpose compilers}
General-purpose quantum programming frameworks, such as \emph{Qiskit}~\cite{qiskit,qiskit_isometry} and \emph{Qualtran}~\cite{berry2019qubitization,babbush2018encoding}, include state-preparation modules that compile arbitrary target states, defined by their amplitude vectors, into executable quantum circuits.
These methods are universal and guarantee exact or approximate state preparation, but current implementations either struggle or are not designed to recognize and exploit structure in the input description. This shortcoming requires users to design custom routines to obtain better circuits. 
This limitation is practical rather than fundamental, and future compilers could in principle identify such structure to reduce resource costs automatically. 
In its present form, the complexity of \emph{Qiskit} scales with the full Hilbert-space dimension, making the method suitable mainly for small systems or proof-of-concept studies. 
On the other hand, \emph{Qualtran} implements a more recent technique that introduces ancillae and performs better on sparse states, which we discuss in the next paragraph.

\paragraph{Sparse states}   
Sparse quantum states with only a small number $\spar \ll 2^N$ of nonzero amplitudes can be compiled far more efficiently than arbitrary dense states. 
This advantage stems not only from the reduced classical description, requiring specification of only $\spar$ nonzero amplitudes, but also from the simpler quantum structure, which enables circuits with resource cost scaling in $\spar$ rather than the full Hilbert-space dimension $2^N$. 
For example, \emph{Qualtran} includes a dedicated subroutine for $\spar$-sparse states~\cite{qualtran, Plesch2011}, which prepares approximate states up to a chosen precision using $O(N\spar)$ gates and a number of ancillae that scales linearly with the system size $N$.
Other notable examples are \citet{gleinig2021efficient}, who proposed an algorithm that prepares $N$-qubit, $\spar$-sparse states using $O(N\spar)$ gates and no ancillary qubits.
More recently, \citet{luo2024circuit} considered regimes with no ancillae, with a sublinear number of ancillae compared to system size $N$, and with unbounded ancillae, showing that all cases achieve circuit sizes scaling as $O(N\spar)$ up to improvements of polylogarithmic factors.
Subsequent works explored sparsity combined with further structural constraints. 
In particular, \citet{de2022double} and \citet{farias2025quantum} studied states supported on low-Hamming-weight basis vectors.
Introducing the concept of \emph{double sparsity}, they refer to states that are sparse both in the number of amplitudes $\spar$ and in their Hamming weight $k$, achieving $O(\spar k)$ gate complexity without ancillae.

Together, these results show that exploiting sparsity and, when present, additional structure, enables efficient compilation well beyond the reach of general-purpose routines, while using a generic list (or dictionary) of amplitudes to describe the desired state.

\paragraph{Specialized routines}
Quantum states possessing combinatorial patterns or underlying symmetry can often be prepared much more efficiently than arbitrary states, as their structure restricts the relevant subspace of the Hilbert space. 
Researchers have leveraged these properties to design efficient circuits for physically relevant families, such as Motzkin~\cite{movassagh2016supercritical}, $W$~\cite{Dur_2000_Wstate}, and Dicke~\cite{Dicke_1954} states. 
Among the most relevant examples, \citet{gopalakrishnan2025pushdown} designed quantum circuits for Motzkin states leveraging concepts from push-down automata.
In our work, one of the families of states we highlight is Dicke-$k$ states, which are uniform superpositions of strings of fixed Hamming weight $k$.
In particular, we will compare our preparation method against \citet{bartschi2019deterministic}, who introduced a general construction using controlled two-qubit rotations with gate complexity \(O(Nk)\) and depth \(O(N)\) without ancillae.
\citet{bartschi2022short} later improved their result by exploiting all-to-all hardware connectivity and achieve logarithmic depth reductions, similar to our \emph{TreeRLSP} method. 
These works exemplify how exploiting structure of \emph{special families of states} enables scalable, hardware-aware compilation beyond the capabilities of general-purpose methods.
However, the \emph{limitation} of specialized routines, is that the user must \emph{explicitly call or implement} them, unless the compiler has \emph{effective user interfaces} and \emph{tests} that \emph{automatically recognize structure} in the user's description. 

\subsection{Regular Language States}
We now introduce the concept of RLS, which are the focus of our work.
Given an alphabet $\Sigma$, we denote by $\Sigma^*$ the set of all strings of arbitrary length formed by concatenating symbols in $\Sigma$. 
A \emph{language} $L$ is any subset of $\Sigma^*$. 
Languages are often defined via rules that the strings need to satisfy. 
In particular, RLs over $\Sigma$ are those defined by regex, which are recursively constructed from the following elements: the empty set $\emptyset$ and the empty string $\varepsilon$ (which differ in the fact that $L\emptyset = \emptyset L = \emptyset$ while $L\varepsilon = \varepsilon L = L$ for any language $L$), single letters $a \in \Sigma$, and the operations of concatenation ($R_1 R_2$), union ($R_1 \cup R_2$) and Kleene star ($R_1^* := \varepsilon \cup R_1 \cup R_1 R_1 \cup \dots$). 
For $N \in \mathbb{N}$, $\Sigma^N$ denotes the set of all words $w$ of length $N$, and any language $L \subseteq \Sigma^N$ is finite and hence regular~\cite{Hopcroft_2001_book_automata}. 

\begin{definition}[Regular language states]
\label{def: RLS}
    Any RL $L$ can be mapped to a family of quantum RLS, denoted $\{\ket{L_N}\}_{N \in \mathbb{N}}$, where each $\ket{L_N}$ is the uniform superposition over all words $w$ in $L$ of length $N$, with physical dimension $d = |\Sigma|$:
    \begin{equation}
        \ket{L_N} = \frac{1}{\sqrt{|L \cap \Sigma^N|}}\sum_{w \in L \cap \Sigma^N} \ket{w}. 
    \end{equation}
\end{definition}

Many of the physically relevant states we have mentioned previously are RLS, such as GHZ states~\cite{Greenberger_1989} ($0^* \cup 1^* \cup \dots$), the W-state~\cite{Dur_2000_Wstate} ($0^* 1 0^*$), the Dicke states~\cite{Dicke_1954} ($0^* (1 0^*)^k$), and any uniform superposition of finite sets of words.

\subsection{Finite-State Automata}
A finite-state automaton is a machine that can decide whether a given word belongs to a specific regular language.
RLs admit equivalent descriptions via regular expressions and finite automata, a key result in automata theory known as Kleene's theorem~\cite{Kleene_1956}.
This equivalence gives RLs an algorithmically tractable representation, enabling a concise description and efficient manipulations.
The following is a formal definition of a DFA.

\begin{definition}[Deterministic finite automaton]
    A \textit{DFA} is defined as $\mathcal{F} = \langle Q, \Sigma, \delta, r_0, F\rangle$, where $Q$ is the set of internal states, $r_0 \in Q$ is the initial state, $F \subseteq Q$ is the set of accepting states, and $\delta : Q \times \Sigma \to Q \cup \{\emptyset\}$ is the transition function. 
    A language $L$ is \textit{accepted} by automaton $\mathcal{F}$ if, for every word $w = x_1 x_2 \dots x_n \in L$, there exists a sequence of states $(q_1, \dots, q_n) \in Q^n$ such that $q_{i+1} = \delta(q_i, x_{i+1})$ for $i = 0, \dots, n-1$, and $q_n \in F$.
\end{definition}

DFAs are commonly represented as graphs where states are nodes, transitions are edges labeled by symbols, the initial state is marked by an incoming arrow, and accepting states are drawn with double circles.
To decide whether a word belongs to the language, the DFA reads one symbol at a time, follows the corresponding transitions (moving to a designated and implicit error state $\emptyset$ if no transition exists), and accepts if and only if the run ends in an accepting state.
\begin{tcolorbox}[
  enhanced,
  breakable,
  colback=gray!3,
  colframe=keynoteblue,
  title={Example 1: Hamming weight $k$ strings/Dicke-$k$ RLS},
  boxrule=0.6pt
]
Consider the RL of words with exactly $3$ ones (Hamming weight $3$): $0^*(10^*)^3$. 
This also corresponds to the Dicke-$3$ state.
Both can be efficiently represented by this DFA
\begin{equation}
    \mathcal{F}: \hspace{-2mm}\begin{tikzpicture}[scale=.2, baseline={([yshift=-3ex]current bounding box.center)}, thick]
        \tikzstyle{small state} = [state, minimum size=0pt, fill=purple,node distance=1.5cm,initial distance=1.5cm,initial text=$ $]
        \tikzset{->}
        \node[small state, initial] (q1) {\small $q_1$};
        \node[small state, right of=q1] (q2) {\small $q_2$};
        \node[small state, right of=q2] (q3) {\small $q_3$};
        \node[small state, accepting, right of=q3] (q4) {\small $q_4$};        
        \draw (q1) edge[loop above] node{\scriptsize $0$} (q1);
        \draw (q2) edge[loop above] node{\scriptsize $0$} (q2);
        \draw (q1) edge[above] node{\scriptsize $1$} (q2);
        \draw (q2) edge[above] node{\scriptsize $1$} (q3);
        \draw (q3) edge[loop above] node{\scriptsize $0$} (q3);
        \draw (q3) edge[above] node{\scriptsize $1$} (q4);
        \draw (q4) edge[loop above] node{\scriptsize $0$} (q4);
    \end{tikzpicture} \ .
\end{equation}
\end{tcolorbox}

Languages restricted to strings of fixed length $N$ can be recognized by directed and acyclic graphs (DAG), very much like a \emph{trie} data structure store words.
We can describe such graphs as having $N+1$ layers, with $N$ sets of edges connecting one layer to the next (one set per each symbol position).  
This gives rise to a specific class of DFAs that we leverage in some parts of our framework, and we call them DAG-DFAs.

\begin{definition}[Directed acyclic DFA]
\label{def: dag-dfa}
    A \textit{DAG-DFA} is $\mathcal{F}_{\mathrm{DAG}} = \langle \{Q^{(i)}\}_{i=0}^N, \Sigma, \{\delta^{(i)}\}_{i=1}^N \rangle$, where $Q^{(0)} = \{q_0\}$ and $Q^{(N)} = \{q_N\}$ are the unique and the accepting states, respectively, and each $\delta^{(i)}: Q^{(i-1)} \times \Sigma \to Q^{(i)} \cup \{\emptyset\}$ defines the transitions between consecutive layers. 
    A language $L \subseteq \Sigma^N$ is accepted by $\mathcal{F}_{\mathrm{DAG}}$ if, for every word $w = x_1 \dots x_N \in L$, there exists a sequence of states $(q_0, \dots, q_N)$ such that $q_i \in Q^{(i)}$ and $q_i = \delta^{(i)}(q_{i-1}, x_i)$, for each $i = 1, \dots, N$. 
\end{definition}

Generalizations of DAG-DFAs to finite languages of varying length are also known as deterministic acyclic finite state automaton (DAFSA) or directed acyclic word graph (DAWG). 
Restricting to words of exactly $N$ symbols, ensures that (1) there cannot be transitions between states of non-consecutive layers, and (2) there is a unique final state, at the $N+1$ layer.

\subsection{Matrix Product States} \label{sec:background_mps}

\paragraph{Tensor networks} TNs are powerful representations of quantum many-body states, expressing highly entangled wavefunctions through networks of low-rank tensors that encode correlations locally. An $N$-qudit state $\ket{\psi}$ is a rank-$N$ tensor requiring $O(d^N)$ amplitudes, but TNs factorize it into smaller tensors with internal (bond) dimensions bounded by some parameter $D$. 
When $D$ remains finite, the number of parameters scales as $O(N d D^2)$, yielding an exponentially more compact description for states with limited entanglement.

\paragraph{TN Diagrammatic notation.}
TNs are conveniently expressed through a \emph{diagrammatic language}, where tensors are represented as nodes and their indices as connecting lines. Each leg corresponds to an index, and connecting two legs denotes contraction. A scalar, vector, and matrix are tensors with zero, one, and two legs, respectively; higher-order tensors have multiple legs. For example, a rank-3 tensor $A^i_{\alpha\beta}$ can be viewed as a set of matrices $\{A^i\}_{i=0}^{d-1}$ with physical index $i$ and bond indices $\alpha,\beta$.

\paragraph{Matrix product states.} By contracting such local tensors along their bonds, one obtains a chain-like structure known as a MPS. MPS have been extensively developed in the context of quantum many-body physics \cite{orus_2014_introTNs, Cirac_review_2021}, where they underpin paradigmatic algorithms such as the density matrix renormalization group (DMRG) \cite{Schollwock_2011}, as well as key theoretical results like the classification of gapped phases of matter in one-dimensional systems \cite{Chen_2011_symm, Schuch_2011_symm}. They are also widely used by the data modeling community, where they are known as \textit{tensor trains} \cite{oseledets_2011_tensor-trains}. 
Throughout this work, we consider \emph{open boundary conditions} (OBC), where the first and last tensors have a single virtual index; i.e., $D_0=D_N=1$ in the following definition.

\begin{definition}[OBC Matrix product states]
\label{def: MPS}
    An MPS on $N$ sites is defined by a set of site-dependent matrices $\{A_{[n]}^i\}_{i=0}^{d-1}$, where each $A_{[n]}^i \in \mathbb{C}^{D_{n-1} \times D_{n}}$ for $n \in \{1, \dots, N\}$ and $D_0 = D_N = 1$.
    These give rise to the quantum state
    \begin{align*}
        \ket{\psi} &= \sum_{i_1 \dots i_N=0}^{d-1} A_{[1]}^{i_1} A_{[2]}^{i_2} \dots A_{[N]}^{i_N}  \ket{i_1 i_2 \dots i_N}
        \equiv 
        \begin{tikzpicture}[scale=.55, baseline={([yshift=-1ex]current bounding box.center)}, thick]
            \begin{scope}[shift={(0,0)}]
                \draw[shift={(0,0)},dotted] (0.5,0) -- (4,0);
                \LeftMPSTensor{(0,0)}{$A_{[1]}$}{purple}
                \MPSTensor{1.5,0}{$A_{[2]}$}{purple}
                \RightMPSTensor{4.5,0}{$A_{[N]}$}{purple}
            \end{scope}
        \end{tikzpicture}
        ,
    \end{align*}
    where $d$ is the \textit{physical dimension}, and $D_n$ are the \textit{bond dimensions} of the internal matrices. 
    We call $D = \max_{n} (D_n)$ the \emph{bond dimension} of the MPS.
\end{definition}

There can be equivalent MPSs, with different bond dimensions, that represent the same state.
In general, the optimal (minimal) bond dimension $D_n$ of a tensor at site $n$, required to faithfully represent a given state, is equal to the Schmidt rank across the corresponding bipartition of the system. 
Thus, the \emph{optimal} bond dimension of an MPS is given by the \emph{maximal Schmidt rank}. 

\begin{definition}[Schmidt rank]
\label{def:schmidt-rank-bipartition}
Let $\ket{\psi}\in(\mathbb{C}^d)^{\otimes N}$ and fix a cut $n\in\{1,\dots,N-1\}$ inducing the bipartition $L=\{1,\dots,n\}$ and $R=\{n{+}1,\dots,N\}$. 
The \emph{Schmidt rank across the bipartition $L|R$} is $\chi_n=\operatorname{rank}(\Psi_{(n)})$, where $\Psi_{(n)}$ is the $d^n\times d^{N-n}$ matrix obtained by reshaping the amplitudes of $\ket{\psi}$ with respect to the cut $n$.
We denote the \emph{maximal Schmidt rank} of a $\ket{\psi}$ as $\chi = \max_{n}(\chi_n)$.
\end{definition}

The Schmidt rank certifies bipartite entanglement for pure states: across a cut $L|R$, $\chi_n=1$ if and only if $\ket{\psi}$ is a product state (no entanglement), while $\chi_n > 1$ implies entanglement across that cut.
Moreover, for the reduced state $\rho_L$, any Rényi entropy (which, for pure states, quantifies bipartite entanglement) satisfies $S_\alpha(\rho_L) \leq \log(\chi_n)$, with equality when $\alpha=0$ or when the nonzero Schmidt coefficients are uniform.
The resources of our circuits will scale with $N$ and this \emph{intrinsic} quantity.

\paragraph{Connection to formal languages.} MPS admit a natural formulation in terms of finite automata~\cite{crosswhite_2008, florido2024regular}, establishing a correspondence between TNs and RLs. 
This duality has been used to accelerate many-body simulations~\cite{crosswhite_2008} and study the expressiveness of TNs~\cite{florido2024regular}. 
Here, we exploit it to efficiently compile regular descriptions into quantum circuits.

\section{RLS Compilation Pipeline} \label{sec:methodology}
Bridging and extending independent lines of work on RLs and MPS, our methodology is designed in a modular fashion, consisting of a sequence of well-defined steps that can be independently optimized. 
The pipeline is summarized in Fig. \ref{fig:pipeline} and proceeds through the following stages:
\begin{itemize}
    \item[---] \textbf{User input to automaton:} the input RLS is specified either as a regex, a DFA, or a set of strings, and converted into a finite automaton. 
    \item[---] \textbf{Complement of the language (optional):} if requested by the user, the automaton is complemented by exchanging accepting and non-accepting states.
    \item[---] \textbf{Automaton to MPS:} the DFA is minimized and mapped to an MPS.
    \item[---] \textbf{MPS to Isometries:} the MPS is decomposed into isometries acting on a subset of qubits and describing the quantum circuit.
    \item[---] \textbf{Isometries to Quantum Circuit}: the isometries are embedded into unitaries and transpiled into a quantum circuit, using gates from a universal set and accounting for connectivity.
\end{itemize}
The following sections describe each step in detail. 
We describe our methodology tailoring it to states of qubits, \emph{i.e.}, restricting the alphabet to $\Sigma = \{0,1\}$. 
However, the same methods easily extend to states of qudits and larger alphabets.

\subsection{User Input to Automata}
\label{sec:user-to-automata}
A user can specify a state in three possible ways: by providing an explicit set of strings, as a regular expression, or directly as a DFA accepting the language. 
The first step of the compilation pipeline translates the input into a DFA, the first tool of our IR. 
We discuss the three input types.

\subsubsection{Set of strings} \label{subsubsec:explicit_list_strings}
When the user describes the state through a finite set (or list) of $\spar$ distinct strings $L = \{w_1, w_2, \dots, w_\spar\} \subset \Sigma^N$, we can construct a DAG-DFA that accepts them: each string defines a path in the automaton from the initial state to the unique accepting state, with shared prefixes leading to state merging to preserve determinism. 
This procedure requires $O(\spar N)$ operations.

\subsubsection{Regular expression}
When the user describes the state via a regular expression (and a system size $N$), we need to compile the regex into a DFA.
This can be done via standard techniques, for instance by building a non-deterministic automaton (NFA)~\cite{Thompson_1968_regex-to-nfa, Glushkov_1961_regex-to-NFA} and making it deterministic~\cite{rabin1959finite, Hopcroft_2001_book_automata, alfred2007compilers}.
While specifying a regex can be substantially more concise than enumerating the entire state support, making an NFA deterministic can incur costs that scale exponentially with the NFA size, and more informally with the regex description size.
We do not expect this blow-up for regular expressions describing common physical states, but this might happen for particular or adversarial regexes.
In our compile-time analysis (\Cref{sec: costs and guarantees}), we will stress this possibility and begin our analysis from the DFA representation, thereby avoiding the risk of this potential exponential overhead altogether. 

\subsubsection{DFA} In certain cases, a user might want to specify the state directly as a DFA.
This feature would provide users with a \emph{workaround enabling efficient descriptions beyond strictly RLs}.
For instance, the Dyck language of well parenthesized strings is a paramount example of context-free, non-regular language, and corresponds to quantum Motzkin states~\cite{movassagh2016supercritical, gopalakrishnan2025pushdown}.
However, any \emph{finite} language is regular, so its restriction to length $N$-strings can be captured by a DFA of $O(N)$ states, one per each parenthesis left to close.

\begin{definition}[$D_{\mathrm{init}}$ and $D_{\mathrm{minDFA}}$]
\label{def: dinit dminDFA}
    We call $D_{\mathrm{init}}$ the number of internal states of a DFA, or the maximum number of states in a layer $\max_{i \in \{0,\dots,N\}}|Q^{(i)}|$ of a DAG-DFA.
    For a minimal and equivalent automaton, we call this parameter $D_{\mathrm{minDFA}}$.
    In general, it holds that $D_{\mathrm{minDFA}} \leq D_{\mathrm{init}}$.
\end{definition}

\paragraph{DFA minimization}
Each RL admits a canonical form as a unique minimal DFA (in the number of states), up to relabeling of the internal states~\cite{Hopcroft_2001_book_automata}.
A DFA can be minimized in time $O(D_{\mathrm{init}} \log(D_{\mathrm{init}}))$ using Hopcroft's algorithm~\cite{Hopcroft_1971, Hopcroft_2001_book_automata}, while a DAG-DFA in $O(ND_{\mathrm{init}})$~\cite{revuz_1992_min-acylic-linear-time}, preserving the DAG-DFA structure.
The minimal DFA canonicalizes the input and serves as the foundation for the rest of the pipeline, offloading computational effort from the MPSs canonization.

\subsection{Complement Regular Language States} \label{sec:complement}
In some applications, it is useful to prepare a complement RLS (\emph{e.g.}, see \cite[Definition 9]{bellante2025quantum} and \cite{benedetti2025complement}). 
\begin{definition}[Complement RLS] \label{def:complementRLS}
    The family of quantum \textit{complement RLSs} associated to a RL $L$ is the RLS associated to its complement RL, $\bar{L} := \Sigma^* \setminus L$. 
\end{definition}
Importantly, if an RLS has a concise description, its complement does too. 
In our framework, the user can request the synthesis of the complement using the base RLS description and setting a flag that our pipeline efficiently adapts. 
Current compilers (\emph{Qiskit} and \emph{Qualtran}) do not offer this functionality. 
Furthermore, if the base RLS is sparse, its complement will have a dense support, making the compilation using \emph{set of strings} or \emph{dictionary} descriptions intractable for these compilers.
Instead, we create a DFA for the base RLS, as explained in the previous section, and if the complement flag is set, the DFA IR is modified as follows:
\begin{enumerate}
    \item If the automaton is a \emph{DAG-DFA}, it is convenient to restrict the complement to length-$N$ strings (i.e.,\ $\bar{L} \cap \Sigma^N$) and keep a DAG-DFA structure, as this will help in building a better MPS later on.
    This complement DAG-DFA can be computed efficiently, in $O(ND_{\mathrm{minDFA}})$ steps. 
    After locating the first layer with undefined transitions (i.e., $\delta^{(i)}(q_{i-1}, x_i) = \emptyset$), we add a dedicated \emph{sink state} to each subsequent layer. 
    All previously undefined transitions now redirect into the next layer's sink, and each layer's sink transitions to the following layer's sink, for all possible symbols.
    The previous accepting state is removed, with its incoming transitions, and the final layer's sink becomes the unique accepting state. 
    The resulting DAG-DFA recognizes $\bar{L} \cap \Sigma^N$.
    We minimize it again, and, by construction, the final $D_{\mathrm{minDFA}}$ increases by at most $1$.

    \item If the automaton is a generic \emph{DFA}, the complement is obtained by adding an explicit \emph{sink}, redirecting the undefined transitions there, and making non-accepting states accepting and vice versa. 
    A subsequent minimization ensures the result is the unique minimal DFA for $\bar{L}$. 
    Even in this case, the final $D_{\mathrm{minDFA}}$ increases by at most $1$.
\end{enumerate}

\begin{tcolorbox}[
  enhanced,
  breakable,
  colback=gray!3,
  colframe=keynoteblue,
  title={Example 2: $W$-State/one-hot encoded strings DFAs},
  boxrule=0.6pt
]
    Let the user specify the RL $L = 0^* 1 0^*$, whose quantum RLS corresponds to the W-state, and its complement $\overline{L}$. 
    The minimal DFAs accepting these languages are, respectively,
    \begin{equation}
    \label{eq: W-state DFA F1}
        \mathcal{F}_1: \hspace{-2mm}\begin{tikzpicture}[scale=.2, baseline={([yshift=-3ex]current bounding box.center)}, thick]
            \tikzstyle{small state} = [state, minimum size=0pt, fill=purple,node distance=1.5cm,initial distance=1.5cm,initial text=$ $]
            \tikzset{->}
            \node[small state, initial] (q1) {\small $q_1$};
            \node[small state, accepting, right of=q1] (q2) {\small $q_2$};
            \draw (q1) edge[loop above] node{\scriptsize $0$} (q1)
            (q2) edge[loop above] node{\scriptsize $0$} (q2)
            (q1) edge[above] node{\scriptsize $1$} (q2);
        \end{tikzpicture} 
        \qquad     
        \bar{\mathcal{F}}_1: \hspace{-2mm}\begin{tikzpicture}[scale=.2, baseline={([yshift=-3ex]current bounding box.center)}, thick]
        \tikzstyle{small state} = [state, minimum size=0pt, fill=purple,node distance=1.5cm,initial distance=1.5cm,initial text=$ $]
        \tikzset{->}
        \node[small state, initial, accepting] (q1) {\small $q_1$};
        \node[small state, right of=q1] (q2) {\small $q_2$};
        \node[small state, accepting, right of=q2] (qs) {\small $q_s$};
        \draw (q1) edge[loop above] node{\scriptsize $0$} (q1)
              (q1) edge[above] node{\scriptsize $1$} (q2)
              (q2) edge[loop above] node{\scriptsize $0$} (q2)
              (q2) edge[above] node{\scriptsize $1$} (qs)
              (qs) edge[loop above] node{\scriptsize $0,1$} (qs);
    \end{tikzpicture} \ .
    \end{equation}

    If the user provides a set of strings for the length-$3$ W-state and asks for the complement too, our pipeline produces the following $4$-layers minimal DAG-DFAs:
    \begin{equation} 
    \label{eq: W-state DFA F2}
    \begin{tikzpicture}[scale=.25, baseline={([yshift=-3ex]current bounding box.center)}, thick]
    \tikzstyle{small state} = [state, minimum size=0pt, fill=purple,node distance=1.3cm,initial distance=1.3cm,initial text=$\mathcal{F}_2:$,inner sep=0.6pt]
    \tikzset{->}
    \node[small state, initial] (q1_0) {\small $q_1^{(0)}$};
    \node[small state, right of=q1_0] (q1_1) {\small $q_1^{(1)}$};
    \node[small state, right of=q1_1] (q1_2) {\small $q_1^{(2)}$};

    \node[small state, below of=q1_1] (q2_1) {\small $q_2^{(1)}$};
    \node[small state, right of=q2_1] (q2_2) {\small $q_2^{(2)}$};
    
    \node[small state, accepting, right of=q2_2] (q2_3) {\small $q_1^{(3)}$};
    
    \draw[->] (q1_0) -- (q1_1) node[midway, above] {\scriptsize $0$};
    \draw[->] (q1_0) -- (q2_1) node[midway, right] {\scriptsize $1$};
    
    \draw[->] (q1_1) -- (q1_2) node[midway, above] {\scriptsize $0$};
    \draw[->] (q1_1) -- (q2_2) node[midway, right] {\scriptsize $1$};
    
    \draw[->] (q1_2) -- (q2_3) node[midway, right] {\scriptsize $1$};
        
    \draw[->] (q2_1) -- (q2_2) node[midway, above] {\scriptsize $0$};
    \draw[->] (q2_2) -- (q2_3) node[midway, above] {\scriptsize $0$};
    \end{tikzpicture}
    \quad     
    \begin{tikzpicture}[scale=.25, baseline={([yshift=-3ex]current bounding box.center)}, thick]
    \tikzstyle{small state} = [state, minimum size=0pt, fill=purple,node distance=1.3cm,initial distance=1.3cm,initial text=$\bar{\mathcal{F}}_2:$,inner sep=0.6pt]
    \tikzset{->}
    \node[small state, initial] (q1_0) {\small $q_1^{(0)}$};
    \node[small state, right of=q1_0] (q1_1) {\small $q_1^{(1)}$};
    \node[small state, right of=q1_1] (q1_2) {\small $q_1^{(2)}$};

    \node[small state, below of=q1_1] (q2_1) {\small $q_2^{(1)}$};
    \node[small state, right of=q2_1] (q2_2) {\small $q_2^{(2)}$};
    
    \node[small state, below of=q2_2] (q3_2) {\small $q_s^{(2)}$};
    
    \node[small state, accepting, right of=q3_2] (q2_s) {\small $q_s^{(3)}$};
    \draw[->] (q1_0) -- (q1_1) node[midway, above] {\scriptsize $0$};
    \draw[->] (q1_0) -- (q2_1) node[midway, right] {\scriptsize $1$};
    
    \draw[->] (q1_1) -- (q1_2) node[midway, above] {\scriptsize $0$};
    \draw[->] (q1_1) -- (q2_2) node[midway, right] {\scriptsize $1$};
    
    \draw[->] (q1_2) -- (q2_s) node[midway, right] {\scriptsize $0$};
        
    \draw[->] (q2_1) -- (q2_2) node[midway, above] {\scriptsize $0$};
    \draw[->] (q2_2) -- (q2_s) node[midway, above] {\scriptsize $1$};

    \draw[->] (q2_1) -- (q3_2) node[midway, above] {\scriptsize $1$};
    \draw[->] (q3_2) -- (q2_s) node[midway, above] {\scriptsize $0,1$};
    \end{tikzpicture}
    \end{equation}
    Note that $D_{\mathrm{minDFA}}^{\mathcal{F}_1}=D_{\mathrm{minDFA}}^{\mathcal{F}_2}=2$ and $D_{\mathrm{minDFA}}^{\bar{\mathcal{F}}_1}=D_{\mathrm{minDFA}}^{\bar{\mathcal{F}}_2}=3$.
    Our DAG-DFA construction allows recovering a small $D_{\mathrm{minDFA}}$ even when starting from a set of strings description rather than from a compact and more general regex. 
\end{tcolorbox}

\subsection{Automata to MPS} \label{sec:MPS}
The next step in our compilation pipeline requires mapping the minimal DFA to an MPS. 
Intuitively, the sequential structure of a DFA provides a natural backbone for an MPS representation. 
DFA states correspond to bond indices, and input symbols from the alphabet $\Sigma$ correspond to the physical indices. 
In a way, the MPS encodes adjacency matrices of the DFA, with the left boundary encoding the initial state and the right boundary projecting onto accepting states. 

Below we describe two practical constructions, starting from a minimal DFA or DAG-DFA.
In both cases, the MPS will have bond dimension $D = D_{\mathrm{minDFA}}$ (\cref{def: dinit dminDFA}).
\subsubsection{Uniform-bulk construction (DFA)}

This correspondence gives rise to an MPS whose tensors are all the same (i.e., with \textit{uniform bulk}) except for the two boundary vectors.

\begin{lemma}[DFA to uniform-bulk MPS \cite{florido2024regular}] \label{lemma:dfa_to_uniform_MPS}
    Given a DFA $\mathcal{F} = \langle Q, \Sigma, \delta, I, F \rangle$ accepting the RL $L$, an MPS description of the family of quantum RLSs $\{ \ket{L_N} \}_{N \in \mathbb{N}}$ associated to it can be obtained so that
    $\ket{L_N} :=
        \begin{tikzpicture}[scale=.45, baseline={([yshift=-1ex]current bounding box.center)}, thick]
            \FullMPS{0,0}{$A$}{purple}
            \draw[fill=amaranth] (-1.4,0) circle (0.4);
            \node at (-1.4,0) {\scriptsize $v_l$};
            \draw[fill=amaranth] (6.4,0) circle (0.4);
            \node at (6.4,0) {\scriptsize $v_r$};
            \node at (3.2,0.5) {\scriptsize $N$ times};
        \end{tikzpicture}$,
    with boundary vectors and bulk tensors defined by
    \begin{equation} \label{eq:associate_MPS_to_NFA1}
     \begin{tikzpicture}[scale=.45, baseline={([yshift=-0.5ex]current bounding box.center)}, thick]
            \draw (0,0) -- (0.8,0);
            \draw[fill=amaranth] (0,0) circle (0.4);
            \node at (0,0) {\scriptsize $v_l$};
        \end{tikzpicture} := \sum_{i \in I} \bra{i}, 
        \qquad 
        \begin{tikzpicture}[scale=.45, baseline={([yshift=-1ex]current bounding box.center)}, thick]
            \MPSTensor{0,0}{$A$}{purple}
            \node at (-1.2,0) {\scriptsize $i$};
            \node at (1.2,0) {\scriptsize $j$};
            \node at (0,1.25) {\scriptsize $x$};
        \end{tikzpicture} :=
        \begin{cases}
            1 & \text{if } j = \delta(i,x), \\
            0 & \text{otherwise},
        \end{cases}
        \qquad 
        \begin{tikzpicture}[scale=.45, baseline={([yshift=-0.5ex]current bounding box.center)}, thick]
            \draw (0,0) -- (-0.8,0);
            \draw[fill=amaranth] (0,0) circle (0.4);
            \node at (0,0) {\scriptsize $v_r$};
        \end{tikzpicture} := \sum_{f \in F} \ket{f}.
    \end{equation}
    This MPS has bond dimension $D = |Q|$.
\end{lemma}
In words, each tensor $A^{x}_{i j} \equiv \{A^x\}_{x \in \Sigma}$ in the bulk encodes the whole adjacency matrix of the DFA graph, with each matrix $A^x$ encoding the edges labeled by the symbol $x$.

\subsubsection{Non-uniform construction (DAG-DFA)}
Using the construction of \cref{lemma:dfa_to_uniform_MPS}~\cite{florido2024regular}, the resulting MPS bond dimension would scale with the total number of internal states.
However, the layered structure of a DAG-DFA can be leveraged to obtain a different, non-uniform MPS construction \emph{with bond dimension equal to the maximum number of nodes in a layer}, maintaining $D=D_{\mathrm{minDFA}}$.

\begin{lemma}[DAG-DFA to non-uniform MPS] \label{lemma:dag-dfa-to-mps}
    Given a DAG-DFA $\mathcal{F} = \langle \{Q^{(i)}\}_{i=0}^{N}, \Sigma, \{\delta^{(i)}\}_{i=0}^{N} \rangle$ accepting the finite language $L \subseteq \Sigma^N$, an MPS description of the associated family of quantum RLS $\{ \ket{L_N} \}_{N \in \mathbb{N}}$ can be obtained so that $\ket{L_N} :=
        \begin{tikzpicture}[scale=.45, baseline={([yshift=-1ex]current bounding box.center)}, thick]
            \LeftMPSTensorRect{0,0}{$A_{[1]}$}{purple}{0.1}
            \MPSTensorRect{1.7,0}{$A_{[2]}$}{purple}{0.1}
            \RightMPSTensorRect{5.4,0}{$A_{[N]}$}{purple}{0.125}
            \node at (3.6,0) {$\dots$};
        \end{tikzpicture},$
    with tensors defined by
    \begin{equation} \label{eq:mps-dag-dfa-correspondence}
        \begin{tikzpicture}[scale=.45, baseline={([yshift=-1ex]current bounding box.center)}, thick]
            \MPSTensorRect{0,0}{$A_{[n]}$}{purple}{0.125}
            \node at (-1.4,0) {\scriptsize $i$};
            \node at (1.4,0) {\scriptsize $j$};
            \node at (0,1.25) {\scriptsize $x$};
        \end{tikzpicture} = 
        \begin{cases}
            1 & \text{if } j = \delta^{(n)}(i, x), \\
            0 & \text{otherwise.}
        \end{cases}
    \end{equation}
    Each $A_{[n]}^x$ is $D_{n-1} \times D_{n}$, with $D_n = |Q^{(n)}|$.
    This MPS has bond dimension $D = \max_n (|Q^{(n)}|)$.
\end{lemma}
\begin{proof}
    The correctness follows from the definitions of RLS~(\ref{def: RLS}), MPS~(\ref{def: MPS}), and DAG-DFA~(\ref{def: dag-dfa}). 
    In particular, observe that for any word $w = x_1 \dots x_N$, the DFA has an accepting path $(q_0, \dots, q_N)$ if and only if $A_{[1]}^{x_1}A_{[2]}^{x_2}\dots A_{[N]}^{x_N} = 1$, or equivalently if $A_{[n]}^{x_n}=1, \forall n \in \{1,\dots,N\}$. 
    Indeed, by \eqref{eq:mps-dag-dfa-correspondence}, each $A_{[n]}^{x_n}$ can be seen as the indicator function $\Sigma \times Q^{(n-1)} \times Q^{(n)} \to \{0,1\}$ of the graph of $\delta^{(n)}$.
    Or, equivalently, $A^{x}_{[n]}$ encodes the adjacency matrix of the bipartite graph between layers $n-1$ and $n$ of the DAG-DFA, restricted to edges for the symbol $x$.
    This implies $D_n = |Q^{(n)}|$ too.
\end{proof}

\begin{tcolorbox}[
  enhanced,
  breakable,
  colback=gray!3,
  colframe=keynoteblue,
  title={Example 3: $W$-State/one-hot encoded strings MPSs},
  boxrule=0.6pt,
]
    Applying the construction of \cref{lemma:dfa_to_uniform_MPS} to the DFA $\mathcal{F}_1$ \eqref{eq: W-state DFA F1} yields the following MPS:
\begin{equation}
\begin{cases}
    \begin{tikzpicture}[scale=.45, baseline={([yshift=-1ex]current bounding box.center)}, thick]
        \MPSTensor{0,0}{$A$}{purple}
        \node at (0,1.3) {\scriptsize $0$};
    \end{tikzpicture}
    = {\begin{pmatrix}
        1 & 0 \\ 0 & 1
    \end{pmatrix}}
    \ , \quad &
    \begin{tikzpicture}[scale=.45, baseline={([yshift=-1ex]current bounding box.center)}, thick]
        \MPSTensor{0,0}{$A$}{purple}
        \node at (0,1.3) {\scriptsize $1$};
    \end{tikzpicture}
    = {\begin{pmatrix}
        0 & 1 \\ 0 & 0
    \end{pmatrix}} \ , \\
    \begin{tikzpicture}[scale=.45, baseline={([yshift=-0.5ex]current bounding box.center)}, thick]
        \draw (0,0) -- (0.8,0);
        \draw[fill=amaranth] (0,0) circle (0.4);
        \node at (0,0) {\scriptsize $v_l$};
    \end{tikzpicture} = (1, 0) \ , 
    \quad &
    \begin{tikzpicture}[scale=.45, baseline={([yshift=-0.5ex]current bounding box.center)}, thick]
        \draw (0,0) -- (-0.8,0);
        \draw[fill=amaranth] (0,0) circle (0.4);
        \node at (0,0) {\scriptsize $v_r$};
    \end{tikzpicture} = (0, 1)^T.
\end{cases}
\end{equation}
Using \cref{lemma:dag-dfa-to-mps} on the DAG-DFA $\mathcal{F}_2$ \eqref{eq: W-state DFA F2}, produces an equivalent MPS, with the same bond dimension $D=2$.
Starting from $\bar{\mathcal{F}}_1$ and $\bar{\mathcal{F}}_2$ one can obtain a MPS for the complement.
\end{tcolorbox}

\subsection{MPS to Isometries} 
\label{sec:MPS-to-isometries}
This step establishes the beginning of the backend, as the remaining effort is to synthesize the MPS IR into a quantum circuit.
We outline two ways, a \emph{sequential} generation scheme~\cite{schon_2005_sequentialMPS} called \SeqRLSP, which is suited for LNN connectivity, and a \emph{tree}/parallel scheme~\cite{Cirac_2009_trees1, malz2024preparation} called \TreeRLSP, which works best on all-to-all devices.
Both methods first turn the MPS into a network of isometries.
To be precise, a local tensor $A_{[n]}^x$ is left isometric if $\sum_x A_{[n]}^{x\dagger}A_{[n]}^x = I_{D_n}$ and right isometric if $\sum_x A_{[n]}^x A_{[n]}^{x\dagger} = I_{D_{n-1}}$.

We start from an MPS of bond dimension $D_{\mathrm{minDFA}}$ of the form 
$\begin{tikzpicture}[scale=.45, baseline={([yshift=-1ex]current bounding box.center)}, thick]
        \begin{scope}[shift={(0,0)}]
            \LeftMPSTensorRect{-6,0}{$A_{[1]}$}{purple}{0.1}
            \node at (-4.25,0) {$\dots$};
            \MPSTensorRect{-2.15,0}{$A_{[N-1]}$}{purple}{0.5}
            \RightMPSTensorRect{0,0}{$A_{[N]}$}{purple}{0.15}
        \end{scope}
    \end{tikzpicture}$,
adapting the uniform-bulk construction of \cref{lemma:dfa_to_uniform_MPS} by setting 
$A_{[1]}^x = \bra{v_l} A^x$, $A_{[N]}^x = A^x \ket{v_r}$, and $A_{[i]}^x = A^x \ \text{for each } i \in \{2, \dots, N-1\}$.
The isometrization procedure depends on the chosen scheme.

\subsubsection{Sequential SVD pass (SeqRLSP)} 
\label{sec: sequential svd}
In the sequential approach, we convert the MPS into a canonical form through a double SVD sweep along the chain of tensors, first from right-to-left, and then left-to-right, as illustrated in \Cref{fig:seqSVD}. 
This double-sweep SVD pass has two jobs: (1) it produces isometries, and (2) it eliminates the redundant dimensions from the MPS tensors, making the local bond dimensions match the Schmidt rank $\chi_n$ across each cut.
\begin{figure}[h!]
    \centering
    \includegraphics[width=\linewidth]{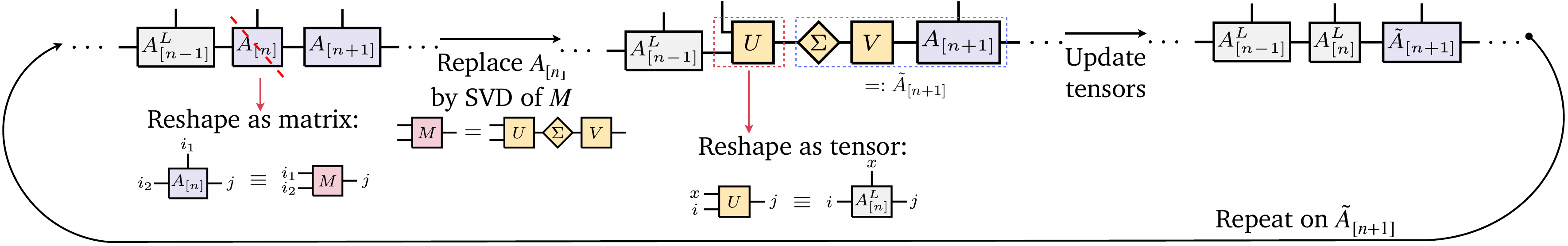}
    \caption{Illustration of the left-to-right sequential SVD procedure. A right-to-left SVD sweep is performed analogously, by mirroring the procedure from the opposite boundary. Each resulting tensor $A_{[n]}^L$ is an isometry of size $\chi_{n-1} \times \chi_n$, with $\chi_n \leq \min \{\chi, 2^n, 2^{N-n}\}$, for $n \in \{1, N-1\}$. It holds that $D_0 = D_N = 1$.}
    \label{fig:seqSVD}
\end{figure}

Since $\chi_n$ is the rank of the $2^n \times 2^{N-n}$ matrix $\Psi^{(n)}$ obtained by reshaping the amplitudes of $\ket{\psi}$ on th $n$-th cut, we have $\chi_n \leq 2^{\min\{n, N-n\}}$.
This property guarantees that the bond dimensions around the boundaries are small, enabling qubit reuse and no ancillae beyond system size (see \cref{fig:seqRLSP circuit}). 
Overall, the bond dimension of this canonical MPS is optimal, equal to the maximal Schmidt rank $D=\chi$. 
This pass eliminates any discrepancy between $D_{\mathrm{minDFA}}$ and $\chi$.

The procedure requires performing a total of $2N$ SVDs. Since computing the SVD of an $m \times n$ matrix costs $O(\min(mn^2, m^2n))$, and the intermediate dimensions throughout the algorithm are upper bounded by $D_{\mathrm{minDFA}}$, the total cost of this pass is $O(ND_{\mathrm{minDFA}}^3)$.

\subsubsection{Tree-based SVD pass (TreeRLSP)}
\label{sec: tree-based SVD}
In the tree-based approach, the canonical MPS produced by the sequential SVD pass above is postprocessed into a tree TN with $O(\log N)$ layers and $O(N)$ isometries~\cite{Shi_2006_tree-TNs, Cirac_2009_trees1, malz2024preparation} via a hierarchical, layer-by-layer SVD scheme (\cref{fig:treeSVD}). 
At each step, the uppermost vertical bonds in the tree carry the physical dimension ($2$ for qubits), the horizontal bonds are bounded by $\chi$, and all remaining vertical bonds by $\chi^2$. 
When reshaping two tensors into the matrix $M_{[n]}$, each of the two upper vertical bonds is bounded by $\chi^2$ and each of the lower ones by $\chi$.
Consequently, the pass requires $O(N)$ SVDs of matrices of size at most $\chi^4 \times \chi^2$, yielding an overall running time of $O(N\chi^8)$.
\begin{figure}[h!]
    \centering
    \includegraphics[width=1\linewidth]{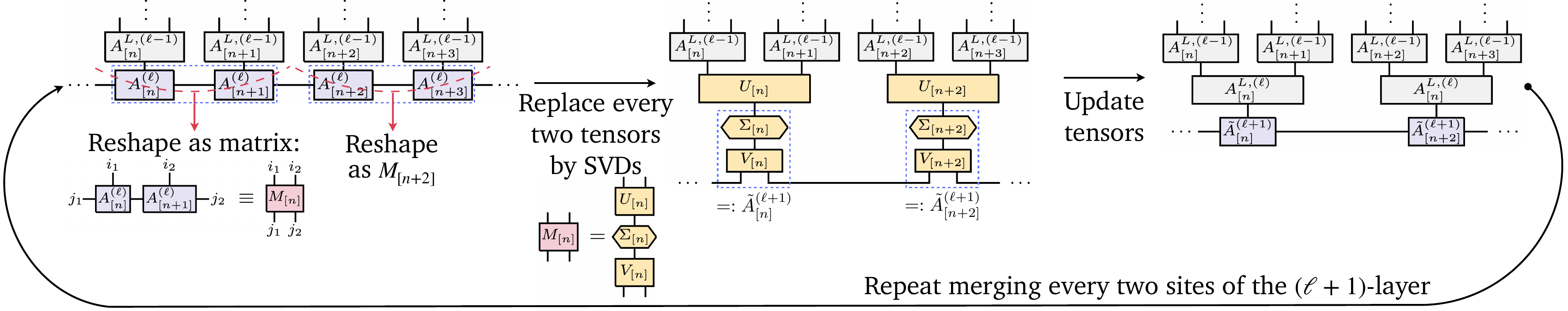}
    \caption{Illustration of the tree-based SVD pass at the $\ell$-th layer. The uppermost vertical bonds in the tree carry the physical dimension ($2$ for qubits), the horizontal bonds are bounded by $\chi$, and all remaining vertical bonds by $\chi^2$. 
    When the number of tensors in a layer is odd, the unpaired one is propagated to the next layer.
    }
    \label{fig:treeSVD}
\end{figure}

\subsection{Isometries to Quantum Circuits}
\label{sec:isometry_to_circuit}
In the previous step, we obtained a TN description of the target RLS. 
The challenge of this last step is to translate the isometric tensors from the TN description into quantum circuit operations and to subsequently transpile them for the target hardware device, considering its universal quantum gate set and connectivity. 

Firstly, since we aim to create a quantum circuit acting on qubits, one needs to appropriately pad the isometric tensors such that each local bond dimension $D_n$ increases to the nearest upper power of $2$, $D_n = 2^{\ceil{\log_2(\chi_n)}} \in O(\chi_n)$. 

Secondly, the padded tensors need to be mapped to quantum circuit operations. This mapping from a TN to quantum circuit operations for both backends can be seen in \eqref{eq:unitary embedding} and \Cref{fig:isometries to circuits}, which clarify how the TN topology affects the required quantum target connectivity. While the \SeqRLSP maps natively to LNN, the operations in \TreeRLSP acts on geometrically non-local qubits, showing why it is more suitable for all-to-all connected devices. 

Thirdly, the quantum operations from \cref{eq:unitary embedding} are converted from a four dimensional tensor form to isometric matrices (two dimensional) by combining the incoming and outgoing legs. One can see that the maximum size of the isometric matrices is $O(2\chi \times \chi)$ for \SeqRLSP, and $O(\chi^4 \times \chi^2)$ for \TreeRLSP. These isometric matrices describe a quantum circuit action mapping an $m$ qubit state, along with $n-m$ qubits initialized to $\ket{0}$, to another $n$ qubit state, $n\geq m$. 
\vspace{-0.2cm}
\begin{equation}
\label{eq:unitary embedding}
    \text{SeqRLSP :}
    {}
    \begin{tikzpicture}[scale=.45, baseline={([yshift=-0.5ex]current bounding box.center)}, thick]
        \begin{scope}[shift={(0,0)}]
            \draw (-0.7,0) -- (1.7,0);
            \draw (-0.7,1) -- (1.7,1);
            \filldraw[fill=yellow] (-0.2,-0.2) -- (-0.2,1.2) -- (1.2,1.2) -- (1.2,-0.2) -- (-0.2,-0.2);
            \draw (0.5,0.5) node {\scriptsize $U_{[n]}$};
            \draw (-1,1) node {\scriptsize $\alpha$};
            \draw (-1.1,0) node {\scriptsize $\bra{0}$};
            \draw (1.9,1) node {\scriptsize $i$};
            \draw (1.9,0) node {\scriptsize $\beta$};
        \end{scope}
    \end{tikzpicture}
    = \begin{tikzpicture}[scale=.45, baseline={([yshift=-1.3ex]current bounding box.center)}, thick]
        \begin{scope}[shift={(0,0)}]
            \MPSTensorRect{(0,0)}{$A^L_{[n]}$}{gray!10}{0.1}
            \draw (-1.3,0) node {\scriptsize $\alpha$};
            \draw (1.3,0) node {\scriptsize $\beta$};
            \draw (0,1.3) node {\scriptsize $i$};
        \end{scope}
    \end{tikzpicture}
    ,
    \quad 
    \text{TreeRLSP :} 
    \begin{tikzpicture}[scale=.45, baseline={([yshift=-0.5ex]current bounding box.center)}, thick]
        \begin{scope}[shift={(0,0)}]
            \draw (-0.7,0) -- (1.7,0);
            \draw (-0.7,1) -- (1.7,1);
            \filldraw[fill=yellow] (-0.2,-0.2) -- (-0.2,1.2) -- (1.2,1.2) -- (1.2,-0.2) -- (-0.2,-0.2);
            \draw (0.5,0.5) node {\scriptsize $U_{[n]}^{(\ell)}$};
            \draw (-1,1) node {\scriptsize $\alpha$};
            \draw (-1.1,0) node {\scriptsize $\bra{0}$};
            \draw (1.9,1) node {\scriptsize $\beta$};
            \draw (1.9,0) node {\scriptsize $\gamma$};
        \end{scope}
    \end{tikzpicture}
    = \begin{tikzpicture}[scale=.45, baseline={([yshift=-0.5ex]current bounding box.center)}, thick]
        \begin{scope}[shift={(0,0)}]
            \draw (0,0) -- (1.7,0);
            \draw (0,1) -- (1.7,1);
            \draw (-0.7,0.5) -- (0,0.5);
            \filldraw[fill=gray!10] (-0.25,-0.2) -- (-0.25,1.2) -- (1.2,1.2) -- (1.2,-0.2) -- (-0.25,-0.2);
            \draw (0.5,0.5) node {\scriptsize $A_{[n]}^{L,(\ell)}$};
            \draw (-1,0.5) node {\scriptsize $\alpha$};
            \draw (1.9,1) node {\scriptsize $\beta$};
            \draw (1.9,0) node {\scriptsize $\gamma$};
        \end{scope}
    \end{tikzpicture}.
\end{equation}

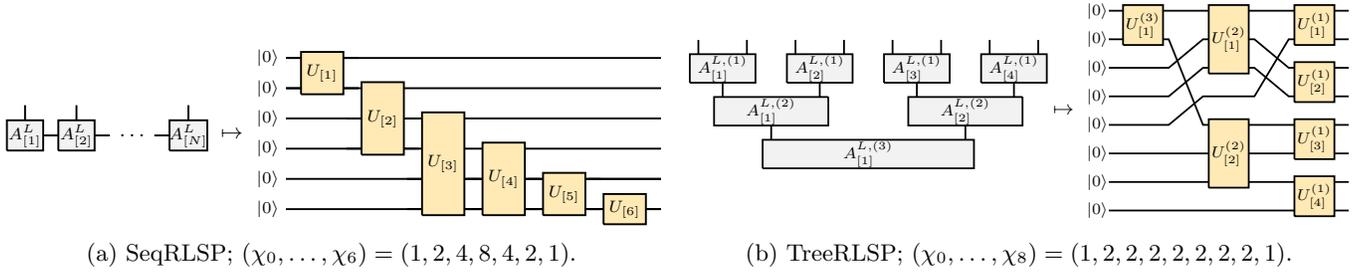
\begin{figure}[t]
\label{fig:tensor_to_circuit}
    \begin{subfigure}{0.49\textwidth}
        \centering
        \resizebox{\linewidth}{!}{$\displaystyle
            \begin{tikzpicture}[scale=.45, baseline={([yshift=-1ex]current bounding box.center)}, thick]
                \LeftMPSTensorRect{0,0}{$A_{[1]}^L$}{gray!10}{0.1}
                \MPSTensorRect{1.7,0}{$A_{[2]}^L$}{gray!10}{0.1}
                \RightMPSTensorRect{5.4,0}{$A_{[N]}^L$}{gray!10}{0.125}
                \node at (3.6,0) {$\dots$};
            \end{tikzpicture} 
            \mapsto
            \begin{tikzpicture}[scale=.45, baseline={([yshift=-0.5ex]current bounding box.center)}, thick]
                \foreach \x in {0,-1,...,-5} {
                    \draw (-1.3,\x) node {\scriptsize $\ket{0}$};
                    \draw (-0.7,\x) -- (11.7,\x);
                }
                \begin{scope}[shift={(0,-1)}]
                    \draw (-0.7,0) -- (1.7,0);
                    \draw (-0.7,1) -- (1.7,1);
                    \filldraw[fill=yellow] (-0.2,-0.2) -- (-0.2,1.2) -- (1.2,1.2) -- (1.2,-0.2) -- (-0.2,-0.2);
                    \draw (0.5,0.5) node {\scriptsize $U_{[1]}$};
                \end{scope}
                \begin{scope}[shift={(2,-3)}]
                    \draw (-0.7,0) -- (1.7,0);
                    \draw (-0.7,1) -- (1.7,1);
                    \filldraw[fill=yellow] (-0.2,-0.2) -- (-0.2,1.2+1) -- (1.2,1.2+1) -- (1.2,-0.2) -- (-0.2,-0.2);
                    \draw (0.5,0.5+0.5) node {\scriptsize $U_{[2]}$};
                \end{scope}
                \begin{scope}[shift={(4,-5)}]
                    \draw (-0.7,0) -- (1.7,0);
                    \draw (-0.7,1) -- (1.7,1);
                    \filldraw[fill=yellow] (-0.2,-0.2) -- (-0.2,1.2+2) -- (1.2,1.2+2) -- (1.2,-0.2) -- (-0.2,-0.2);
                    \draw (0.5,0.5+1) node {\scriptsize $U_{[3]}$};
                \end{scope}
                \begin{scope}[shift={(6,-5)}]
                    \draw (-0.7,0) -- (1.7,0);
                    \draw (-0.7,1) -- (1.7,1);
                    \filldraw[fill=yellow] (-0.2,-0.2) -- (-0.2,1.2+1) -- (1.2,1.2+1) -- (1.2,-0.2) -- (-0.2,-0.2);
                    \draw (0.5,0.5+0.5) node {\scriptsize $U_{[4]}$};
                \end{scope}
                \begin{scope}[shift={(8,-5)}]
                    \draw (-0.7,0) -- (1.7,0);
                    \draw (-0.7,1) -- (1.7,1);
                    \filldraw[fill=yellow] (-0.2,-0.2) -- (-0.2,1.2) -- (1.2,1.2) -- (1.2,-0.2) -- (-0.2,-0.2);
                    \draw (0.5,0.5) node {\scriptsize $U_{[5]}$};
                \end{scope}
                \begin{scope}[shift={(10,-5)}]
                    \draw (-0.7,0) -- (1.7,0);
                    \draw (-0.7,1) -- (1.7,1);
                    \filldraw[fill=yellow] (-0.2,-0.5) -- (-0.2,0.5) -- (1.2,0.5) -- (1.2,-0.5) -- (-0.2,-0.5);
                    \draw (0.5,0) node {\scriptsize $U_{[6]}$};
                \end{scope}
            \end{tikzpicture}
        $}
        \caption{\SeqRLSP; $(\chi_0, \dots, \chi_6) = (1, 2, 4, 8, 4, 2, 1)$.}
        \label{fig:seqRLSP circuit}
    \end{subfigure}
    \hfill
    \begin{subfigure}{0.49\textwidth}
    \centering
    \resizebox{\linewidth}{!}{$\displaystyle
        \begin{tikzpicture}[scale=.45, baseline={([yshift=-1ex]current bounding box.center)}, thick]
            \foreach \x in {0,1,...,7} {
                \draw ({1.7*\x},0) -- ({1.7*\x},1);
            }
            \foreach \x/\k in {0/1, 3.4/2, 6.8/3, 10.2/4} {
                \filldraw[fill=gray!10] ({\x-0.3},-0.5) -- ({\x+2},-0.5) -- ({\x+2},0.5) -- ({\x-0.3},0.5) -- ({\x-0.3},-0.5); 
                \draw (\x + 0.85,0) node {\scriptsize $A^{L,(1)}_{[\k]}$};
             }
             \foreach \x/\k in {0/1, 6.8/2} {
                \draw (\x+0.85,-0.5) -- (\x+0.85,-1);
                \begin{scope}[shift={(\x+0.85,-1.5)}]
                    \filldraw[fill=gray!10] (3.7,0.5) -- (3.7,-0.5) -- (-0.3,-0.5) -- (-0.3,0.5) -- (3.7,0.5);
                    \draw (3.4,0.5) -- (3.4,1);
                    \draw (1.7,0) node {\scriptsize $A^{L,(2)}_{[\k]}$};
                \end{scope}
                \begin{scope}[shift={(\x+1.7+0.85,-2.9)}]
                    \draw (0,0.4) -- (0,0.9);
                \end{scope}
             }
             \begin{scope}[shift={(0,-3)}]
                \filldraw[fill=gray!10] (2.25,-0.5) -- (2.25,0.5) -- (9.65,0.5) -- (9.65,-0.5) -- (2.25,-0.5);
                \draw (5.95,0) node {\scriptsize $A^{L,(3)}_{[1]}$};
             \end{scope}
        \end{tikzpicture}
        \mapsto
        \begin{tikzpicture}[scale=.45, baseline={([yshift=-0.5ex]current bounding box.center)}, thick]
            \foreach \x in {0,-1,...,-7} {
                \draw (-3.1,\x) node {\scriptsize $\ket{0}$};
            }
            \draw (-2.7,0) -- (5.7,0);
            \draw (-2.7,-1) -- (-0.6,-1) -- (0.6,-4) -- (5.7,-4);
            \draw (-2.7,-2) -- (-0.6,-2) -- (0.6,-1) -- (2.4,-1) -- (3.6,-2) -- (5.7,-2);
            \draw (-2.7,-3) -- (-0.6,-3) -- (0.6,-2) -- (2.4,-2) -- (3.6,-3) -- (5.7,-3);
            \draw (-2.7,-4) -- (-0.6,-4) -- (0.6,-3) -- (2.4,-3) -- (3.6,-1) -- (5.7,-1);
            \draw (-2.7,-5) -- (5.7,-5);
            \draw (-2.7,-6) -- (5.7,-6);
            \draw (-2.7,-7) -- (5.7,-7);
            \begin{scope}[shift={(-2,-1)}]
                \filldraw[fill=yellow] (-0.2,-0.2) -- (-0.2,1.2+0) -- (1.2,1.2+0) -- (1.2,-0.2) -- (-0.2,-0.2);
                \draw (0.5,0.5+0) node {\scriptsize $U_{[1]}^{(3)}$};
            \end{scope}
            \begin{scope}[shift={(1,-2)}]
                \filldraw[fill=yellow] (-0.2,-0.2) -- (-0.2,1.2+1) -- (1.2,1.2+1) -- (1.2,-0.2) -- (-0.2,-0.2);
                \draw (0.5,0.5+0.5) node {\scriptsize $U_{[1]}^{(2)}$};
            \end{scope}
            \begin{scope}[shift={(1,-6)}]
                \filldraw[fill=yellow] (-0.2,-0.2) -- (-0.2,1.2+1) -- (1.2,1.2+1) -- (1.2,-0.2) -- (-0.2,-0.2);
                \draw (0.5,0.5+0.5) node {\scriptsize $U_{[2]}^{(2)}$};
            \end{scope}
            \begin{scope}[shift={(4,-1)}]
                \filldraw[fill=yellow] (-0.2,-0.2) -- (-0.2,1.2) -- (1.2,1.2) -- (1.2,-0.2) -- (-0.2,-0.2);
                \draw (0.5,0.5) node {\scriptsize $U_{[1]}^{(1)}$};
            \end{scope}
            \begin{scope}[shift={(4,-3)}]
                \filldraw[fill=yellow] (-0.2,-0.2) -- (-0.2,1.2) -- (1.2,1.2) -- (1.2,-0.2) -- (-0.2,-0.2);
                \draw (0.5,0.5) node {\scriptsize $U_{[2]}^{(1)}$};
            \end{scope}
            \begin{scope}[shift={(4,-5)}]
                \filldraw[fill=yellow] (-0.2,-0.2) -- (-0.2,1.2) -- (1.2,1.2) -- (1.2,-0.2) -- (-0.2,-0.2);
                \draw (0.5,0.5) node {\scriptsize $U_{[3]}^{(1)}$};
            \end{scope}
            \begin{scope}[shift={(4,-7)}]
                \filldraw[fill=yellow] (-0.2,-0.2) -- (-0.2,1.2) -- (1.2,1.2) -- (1.2,-0.2) -- (-0.2,-0.2);
                \draw (0.5,0.5) node {\scriptsize $U_{[4]}^{(1)}$};
            \end{scope}
        \end{tikzpicture}
    $}
        \caption{TreeRLSP; $(\chi_0, \dots, \chi_8) = (1, 2, 2, 2, 2, 2, 2, 2, 1)$.}
        \label{fig:treeRLSP circuit}
    \end{subfigure}
    \caption{Visual mapping between the TN and quantum circuit operations. (a) Mapping for \SeqRLSP. (b) Mapping for \TreeRLSP.
    }
    \label{fig:isometries to circuits}
\end{figure}

Fourth, given an $2^n \times 2^m$ isometric matrix, we use a built in \emph{Qiskit} transpilation method~\cite{qiskit, qiskit_isometry} to map it into a circuit over a chosen universal gate set with a total gate ($\CNOT$) cost $O(2^{m+n})$ and using classical pre-processing time $O(n2^{2m+n})$~\cite{qiskit_comp_classical_time}. This leads to the gate cost for implementing a single operation $O(\chi^2)$ for \SeqRLSP and $O(\chi^6)$ for \TreeRLSP. Since there are $O(N)$ of such isometries in both pipelines, \SeqRLSP and \TreeRLSP, the total gate count incurs a factor of $N$. Once the circuit is expressed in terms of the padded isometric matrices, any standard transpiler can be used to complete synthesis, making our pipeline flexible to improvements at this stage. 
This is also where our compiler may incur overhead relative to purpose-built routines (e.g., Dicke/W): we do not analyze the resulting isometries in detail or optimize the MPS passes for elementary which may obscure exploitable structure and hence miss additional resource reductions. 
Such targeted optimizations are beyond the scope of this work.

\section{Compilation Costs and Circuit Guarantees}
\label{sec: costs and guarantees}
In this section, we review, discuss, and theoretically bound the worst-case performance of our pipeline, both in terms of compile time and quantum resources.
While doing so, we justify our design choices.
Furthermore, we tightly relate the Schmidt rank of an RLS to the one of its complement, relating their entanglement and showing that our methods can synthesize both with the same scaling.
\Cref{table:our_scaling} summarizes the scalings of both backends - \SeqRLSP and \TreeRLSP.

\paragraph{Compilation cost}
We study the compilation time performance as a function of the system size $N$, the initial DFA state count $|Q|$, and the minimal DFA size $D_{\mathrm{minDFA}}$ (Def.~\ref{def: dinit dminDFA}).

The frontend cost of building the automaton depends on the user description: (1) for a \emph{set of strings} of size $\spar$, the DAG-DFA construction costs $O(N\spar)$; (2) for a \emph{regex}, determinization can be exponential in the description size (worst case), though regex-to-DFA pipelines are widely effective in practice; (3) for a \emph{user-supplied DFA}, the compiler proceeds directly to minimization.

Starting from an automaton, we can bound compile time as follows.
\begin{theorem}[Compile-time from the automaton]
\label{thm: compile-time dfa}
    Let the target RLS be represented by a DFA with $|Q|$ internal states, and let $D_{\mathrm{minDFA}}$ be as in \cref{def: dinit dminDFA}.
    Then, the circuit synthesis costs $O(|Q|\log(|Q|)+ND_{\mathrm{minDFA}}^3)$ for \SeqRLSP, and $O(|Q|\log(|Q|)+N(D_{\mathrm{minDFA}}^3+\chi^8))$ for \TreeRLSP, up to $O(\log(\chi))$ factors.
    For a DAG-DFA, the $\log(|Q|)$ factors can be removed.
\end{theorem}
\begin{proof}
    Handling the optional complement flag costs $O(|Q|)$ (Sec.~\ref{def:complementRLS}). 
    Hopcroft's minimization runs in $O(|Q|\log(|Q|)$, and in the acyclic (DAG-DFA) case the $\log(|Q|)$ factor drops~(\cref{sec:user-to-automata}). 
    Constructing the MPS from \cref{lemma:dfa_to_uniform_MPS,lemma:dag-dfa-to-mps} can be done in time proportional to the state count of the minimal DFA, which is bounded by $O(ND_{\mathrm{minDFA}})$.
    The sequential SVD pass costs $O(ND_{\mathrm{minDFA}}^3)$ (\cref{sec: sequential svd}), and the tree network one $O(N\chi^8)$ (\cref{sec: tree-based SVD}). 
    Finally the isometry synthesis costs $O(N\chi^3)$ for SeqRLSP and $O(N\chi^8)$ for TreeRLSP, up to $\log(\chi)$ factors~(\cref{sec:isometry_to_circuit}).
\end{proof}

Together with a \emph{set of strings} input, we can characterize the pipeline's end-to-end complexity.
\begin{corollary}[Compile-time from a Set of Strings]
\label{coro: compile-time set of strings}
    Let the target RLS be specified by a \emph{set of strings} of cardinality $\spar$.
    The end-to-end circuit synthesis uses a DAG-DFA and costs $O(N(\spar + D_{\mathrm{minDFA}}^3))$ for \SeqRLSP and $O(N(\spar + D_{\mathrm{minDFA}}^3+\chi^8))$ for \TreeRLSP.
\end{corollary}
\begin{proof}
    The DAG-DFA can be created in time $O(N\spar)$, with as many internal states.
\end{proof}

These bounds highlight the importance of DFA minimization.
While the sequential SVD pass recovers an optimal MPS of bond dimension $\chi$ regardless, skipping minimization replaces $D_{\mathrm{minDFA}}$ by $D_{\mathrm{init}}$ in the $ND^3$ term.
There are examples where $D_{\mathrm{minDFA}}$ is constant in $N$, but $D_{\mathrm{init}}$ is not.
The advantage is especially clear for \emph{set of strings} input: e.g., the $N$-qubit W-State can be given by $\spar=N$ strings, leading to a DAG-DFA with $D_{\mathrm{init}} = N$, yet $D_{\mathrm{minDFA}}=2$ (see Example 1).
DFA minimization thus replaces an $N^3$ contribution by a constant.
This gap grows with larger supports (e.g., Dicke-$k$: $\binom{N}{k}$ vs $k+1$), as we notice that $D_{\mathrm{minDFA}}$ approximates $\chi$ well.

Finally, for the \emph{set of strings} input, one might wonder why use a DFA IR at all. 
The reason is that it gives \emph{efficient ways of building} the initial MPS from sparse descriptions.
Indeed, the textbook TN procedure~\cite[Sec. 4.1.3, Figure 5]{Schollwock_2011} requires performing $N$ exact SVDs on matrices of size up to $2^{N/2} \times 2^{N/2}$, incurring into an exponential overhead even for sparse descriptions. 

\paragraph{Complement Schmidt rank}
Before giving circuit guarantees and resource bounds, we relate the Schmidt ranks of an RLS and its complement.
Despite potentially exponential differences in support size, their bipartite entanglement is tightly linked, implying that our synthesis uses essentially the same quantum resources for either state.

\begin{theorem}[Complement-RLS Schmidt ranks]
    \label{thm:complementRLS}
    Across each cut, the Schmidt ranks of an RLS and its complement differ by at most one. 
    Equivalently, $|\chi_n - \bar{\chi}_n| \leq 1$.
\end{theorem}
\begin{proof}
    Let $\chi_n, \bar{\chi}_n$ denote the Schmidt ranks across the $n$-th cut for an RLS and its complement, respectively. 
    For each $N$, the complement RLS can be written as $\ket{\bar{L}_N} = \ket{+}^{\otimes N} - \ket{L_N}$.
    Hence, its amplitude matrix reshaped across the $n$-th cut is given by $\Psi^{(n)}_{\ket{\bar{L}_N}} = \Psi^{(n)}_{\ket{+}^{\otimes N}} - \Psi^{(n)}_{\ket{L_N}}$, and $\bar{\chi}_n$ is the rank of this matrix. Since $\mathrm{rank}(\Psi^{(n)}_{\ket{L_N}}) = \chi_n$ and $\mathrm{rank}(\Psi^{(n)}_{\ket{+}^{\otimes N}}) = 1$, together with the fact that any matrices $A,B$ satisfy that $|\mathrm{rank}(A) - \mathrm{rank}(B)| \leq \mathrm{rank}(A-B) \leq \mathrm{rank}(A)+\mathrm{rank}(B)$, the claim follows.
\end{proof}

\paragraph{Circuit guarantees and resource count} 
Finally, we discuss the worst-case quantum resources required by the two backends.
For \SeqRLSP, the total cost is dominated by $O(N)$ layers of unitaries, each acting on at most $\lceil \log_2(\chi) \rceil + 1$ neighboring qubits.
\begin{theorem}[\SeqRLSP resources] 
\label{thm:seqRLSP}
    Let $N$ be the system size and let the RLS have maximal Schmidt rank $\chi$. 
    The \SeqRLSP pipeline produces an LNN circuit that prepares the RLS (or its complement) using $O(\chi^2N)$ depth and elementary gates, with no ancillary qubits.
\end{theorem}
\begin{proof}
    The circuit correctness follows from the definition of RLS~(\cref{def: RLS}), the DFA-MPS equivalence (\cref{lemma:dfa_to_uniform_MPS,lemma:dag-dfa-to-mps}), and the MPS sequential prepration protocol~\cite{schon_2005_sequentialMPS}, as detailed in \cref{sec:MPS-to-isometries}. 
    In \SeqRLSP, the circuit is a sequence of $O(N)$ nearest-neighbour isometries, each of size $O(\chi)$~(\cref{sec:MPS-to-isometries}). 
    Each such isometry can be decomposed into $O(\chi^2)$ elementary gates on an LNN architecture~(\cref{sec:isometry_to_circuit}).    
    Because these isometries are applied sequentially along the chain, the total gate count and depth are both $O(\chi^2N)$.
    By \cref{thm:complementRLS}, the complement state has Schmidt ranks that differ by at most one across every cut, so the same resource bounds apply.
\end{proof}

For \TreeRLSP, correctness follows as in \cref{thm:seqRLSP}, replacing the sequential scheme with the tree~(\cref{sec: tree-based SVD} and \cite{Cirac_2009_trees1, malz2024preparation}).
The tree has $O(\log_2 N)$ layers, and the number of isometries halves from one layer to the next, yielding $O(N)$ total isometries.
In the worst case, an isometry acts on $n = \lceil \log_2(\chi^4) \rceil$ qubits, and has dimensions $\chi^4 \times \chi^2$ (\cref{sec: tree-based SVD}). 
Because these isometries may couple non-adjacent qubits, this backend is more suited to all-to-all platforms such as trapped ions~\cite{liu2025certified} or neutral atoms~\cite{bluvstein2024logical}. 
We collect these considerations in the following theorem.

\begin{theorem}[TreeRLSP resources]
\label{thm:treeRLSP}
    Let $N$ be the system size and let the RLS have maximal Schmidt rank $\chi$.
    The \TreeRLSP pipeline produces an all-to-all circuit that prepares the RLS (or its complement) using $O(\chi^6 \log(N))$ depth, $O(\chi^6N)$ elementary gates, and no ancillae. 
\end{theorem}
\begin{proof}
    Proceeds equivalently to \cref{thm:seqRLSP}, with the considerations above.
\end{proof}

\begin{table}[]
    \centering
    \resizebox{\linewidth}{!}{
    \begin{tabular}{|c|c|c|c|c|c|}
    \hline
         Method & Total gates & Depth & Connectivity & Compile time (DFA input) & Compile time (set of $\spar$ strings) \\
         \hline
        \SeqRLSP & $O(\chi^2N)$ & $O\left(\chi^2N\right)$ & Local &  $O\left(|Q|\log(|Q|)+ND_{\mathrm{minDFA}}^3\right)$ & $O(N(\spar + D_{\mathrm{minDFA}}^3))$\\
        \TreeRLSP & $O(\chi^6 N )$ & $O\left(\chi^6 \log(N)\right)$ & All-to-All & $O\left(|Q|\log(|Q|)+N(D_{\mathrm{minDFA}}^3+\chi^8)\right)$ & $O(N(\spar + D_{\mathrm{minDFA}}^3+\chi^8))$ \\
        \hline
    \end{tabular}
    }
    \caption{Resource and compile time scaling of \SeqRLSP and \TreeRLSP for an RLS and its complement. 
    $N$ is the system size, $\chi$ the maximal Schmidt rank~(\cref{def:schmidt-rank-bipartition}), $|Q|$ the initial automaton state count, and $D_{\mathrm{minDFA}}$ as in \Cref{def: dinit dminDFA}.
    Both methods use no ancillae. 
    \quoted{Compile time (DFA input)} excludes the DFA instantiation; \quoted{Compile time (set of $\spar$ strings)} is end-to-end.}
    \label{table:our_scaling}
\end{table}

\section{Numerical Experiments}
\label{sec:experiments}
In this section, we present numerical results characterizing the advantage of our compilation strategy and the resource requirements of our two backends, \SeqRLSP and \TreeRLSP. 
We benchmark them against widely used approaches from the literature, comparing performance across the standard figures of merit for resource count: the circuit depth, the total number of quantum gates, and the classical compilation time required to generate the circuit description. 

The compilation time is measured using 72 cores of an Intel(R) Xeon(R) Platinum 8360Y CPU @ 2.40GHz (Icelake) with 512GB of RAM (limited to 100GB). All numerical experiments are implemented in \texttt{Python}\footnote{Our code is available at \url{https://github.com/reinisirmejs/RLSComp}.}. For RLS optimization we use the \texttt{pyformlang} package~\cite{pyformlang}, and for MPS calculations we use \texttt{NumPy}~\cite{harris2020array}. 
When reporting the cost in terms of depth and gate count, we use a native \emph{Qiskit} function to transpile any given unitary (or isometry) into gates from the universal set $[\CNOT, \mathrm{RZ}(\theta), \sqrt{\mathrm{X}}, \mathrm{X}, \mathbb{I}]$. 
The cost of transpiling matrices to a quantum circuit is included in the compilation time.
When using a \emph{set of strings} (list or dictionary) description, we omit the time that the user needs to create the list.

The methods we benchmark against are established state-preparation approaches implemented in popular quantum computing frameworks and discussed in the literature (see Section \cref{sec:literature_review}). 
We summarize them in this list:
\begin{itemize}
    \item \textbf{Qiskit}~\cite{qiskit, qiskit_isometry}: This general-purpose method scales exponentially with the system size $N$; we therefore restrict our benchmarks to systems of up to $N \leq 16$ qubits.
    \item \textbf{Qualtran}~\cite{qualtran, Plesch2011}: An approximate sparse-state preparation method whose number of ancillae scales with $N$ and whose capability is intrinsically limited to $N \leq 32$. The target relative precision is fixed to $2^{-3}$ throughout.
    \item \textbf{Gleinig and Hoefler}~\cite{gleinig2021efficient}: An exact sparse-state preparation algorithm that requires no ancillae. 
    \item \textbf{B\"artschi and Eidenbenz}~\cite{bartschi2019deterministic}: A specialized routine for the deterministic preparation of Dicke states on LNN hardware.
\end{itemize} 

\subsection{Dicke States}

Dicke states~\cite{Dicke_1954} constitute a class of states that can be prepared efficiently within our framework. A Dicke-$k$ state is the uniform superposition of all bitstrings of length $N$ with Hamming weight $k$. 

In \cref{fig:Dicke3}, we present results for Dicke state preparation with $k=3$ across various system sizes $N$.
We compare \SeqRLSP and \TreeRLSP against \emph{Qiskit} \cite{qiskit}, \emph{Qualtran} \cite{qualtran, Plesch2011}, \citet{gleinig2021efficient}, and \citet{bartschi2019deterministic}. Although this state has sparsity $\binom{N}{k}$, which scales as $O(N^{\min(k, N-k)})$, it admits an exact MPS representation with bond dimension $\chi = k + 1$. This power-law scaling with $k$ poses a challenge for sparse state preparation methods, as the sparsity grows rapidly with $N$. As shown in \cref{fig:Dicke3}(c), the sparse methods exceeded compilation times of 100~seconds for moderate systems with $N \approx 20$. The purpose-built, state of the art, Dicke state compiler by \citet{bartschi2019deterministic} prepares the state in $O(N)$ depth and $O(kN)$ gates on LNN connectivity. The \SeqRLSP compiler reproduces the same scaling in $N$, with total gate count $O(k^2 N)$, since $\chi = k + 1$ (see \cref{table:our_scaling}).

On architectures with all-to-all connectivity, the preparation cost can be reduced to $O(k \log N)$~\cite{bartschi2022short}. Once again, the \TreeRLSP method achieves the same $N$-scaling but less favorable $k$-dependence, with circuit depth $O(k^6 \log N)$.
The suboptimal $k$-scaling in both cases might arise from the use of a general-purpose circuit decomposition backend (see \cref{sec:isometry_to_circuit}).

We observe that \SeqRLSP and \citet{bartschi2019deterministic} exhibit the same scaling with $N$, with \citet{bartschi2019deterministic} outperforming \SeqRLSP by only a small constant factor in circuit depth, gate count, and compilation time.
This demonstrates the effectiveness of our general \SeqRLSP approach, which closely matches the performance of a method specifically designed for Dicke states.
While \TreeRLSP initially incurs higher costs, its $O(\log N)$ depth scaling would eventually surpass the $O(N)$ scaling of \SeqRLSP. We illustrate this behavior further in \cref{sec:w_state}, where we analyze the W-state preparation.

\begin{figure}[t]
    \centering
    \includegraphics[width=\linewidth]{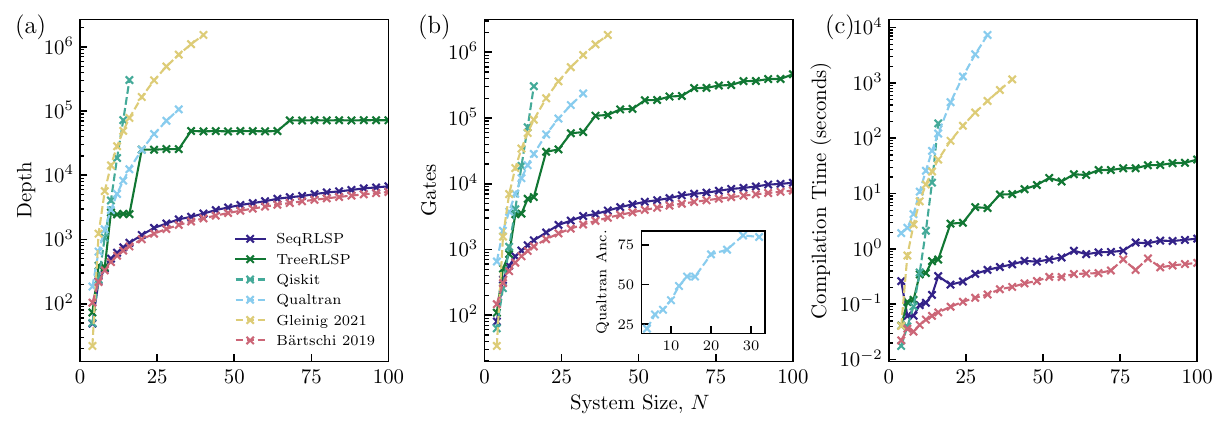}
    \caption{Preparation cost of Dicke$-3$ state for various system sizes $N$. 
    We compare \SeqRLSP, \TreeRLSP, \emph{Qiskit}, \emph{Qualtran}, \citet{gleinig2021efficient}, and \citet{bartschi2019deterministic}. (a) Quantum circuit depth. (b) Total gate count. The inset show the number of ancillae used by \emph{Qualtran}. (c) Compilation time in seconds.}
    \label{fig:Dicke3}
\end{figure}

\begin{figure}[t]
    \centering
    \includegraphics[width=\linewidth]{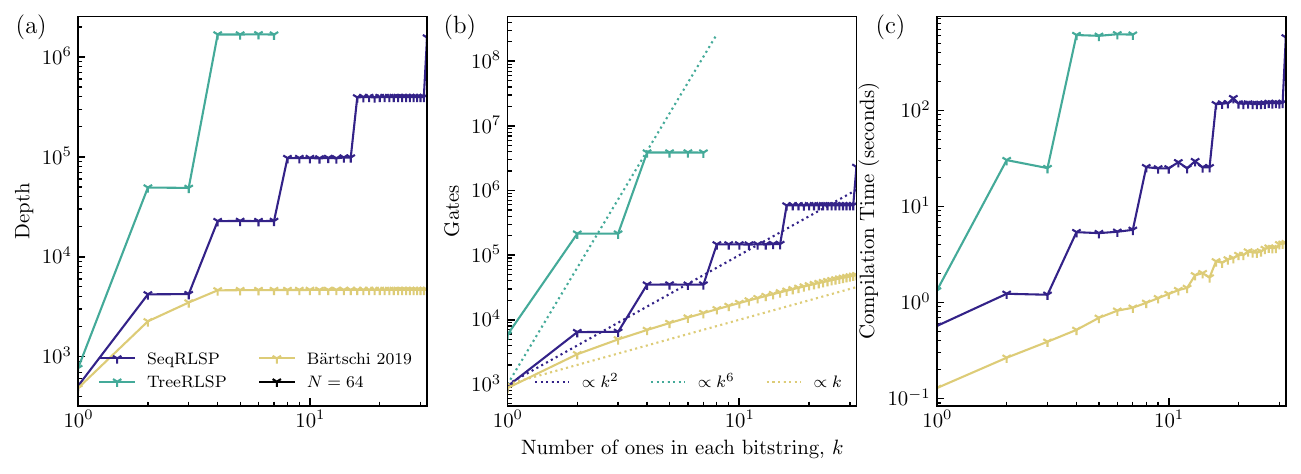}
    \caption{Preparation cost of Dicke$-k$ states on $N=64$ qubits for varying Hamming weight $k$. We compare \SeqRLSP, \TreeRLSP, and \citet{bartschi2019deterministic}. (a) Quantum circuit depth. (b) Total gate count. The dotted trendlines show the expected theoretical $k$-scaling. (c) Compilation time in seconds.}
    \label{fig:dicke_k_scaling}
\end{figure}

In \cref{fig:dicke_k_scaling}, we present results for Dicke state preparation at system size $N=64$ and varying Hamming weight $k$.
We compare the $k$-scaling of \SeqRLSP and \TreeRLSP with \citet{bartschi2019deterministic}. In \cref{fig:dicke_k_scaling}~(b) we numerically verify the $k$-scalings. While \SeqRLSP and \citet{bartschi2019deterministic} show good theoretical agreement, for \TreeRLSP the observed scaling appears more favorable than the expected $k^6$; however, further verification is limited by the large dimensions of the isometries, $(k+1)^4 \times (k+1)^2$, which exceed \emph{Qiskit}’s decomposition capabilities.

A noticeable feature in \cref{fig:dicke_k_scaling} for both \SeqRLSP and \TreeRLSP is the stepwise structure observed in depth and total gate count.
This arises because the isometries are padded to the nearest upper power of two (see \cref{sec:isometry_to_circuit}). Consequently, sudden jumps in circuit depth and gate count occur whenever $\chi = k+1$ exceeds the previous power of two, producing steps at $k = \{2, 4, 8, \dots\}$.

\subsection{W-States}\label{sec:w_state}
The W-state is a special case of a Dicke state with $k=1$~\cite{Dur_2000_Wstate}.
For a system of size $N$, it is an uniform superposition of all $N$ \emph{one-hot} states. In \cref{fig:Wstate}, we report numerical results for W-state preparation across various system sizes $N$, up to $N=256$.
We compare \SeqRLSP and \TreeRLSP against \emph{Qiskit}~\cite{qiskit}, \emph{Qualtran}~\cite{qualtran, Plesch2011}, \citet{gleinig2021efficient}, and \citet{bartschi2019deterministic}.
The low sparsity of the W-state allows \citet{gleinig2021efficient} to reach larger $N$ than in \cref{fig:Dicke3}.
In \cref{fig:Wstate} we observe that the circuit depth and total gate count are very similar among \SeqRLSP, \citet{gleinig2021efficient}, and \citet{bartschi2019deterministic}, while \SeqRLSP consistently achieves the lowest compilation time, even though all methods complete within a few seconds.
For \TreeRLSP, the $O(\log N)$ depth scaling eventually surpasses \SeqRLSP around $N \approx 150$, highlighting its advantage for very large system sizes.
\begin{figure}[t]
\vspace{-0.3cm}
    \centering
    \includegraphics[width=\linewidth]{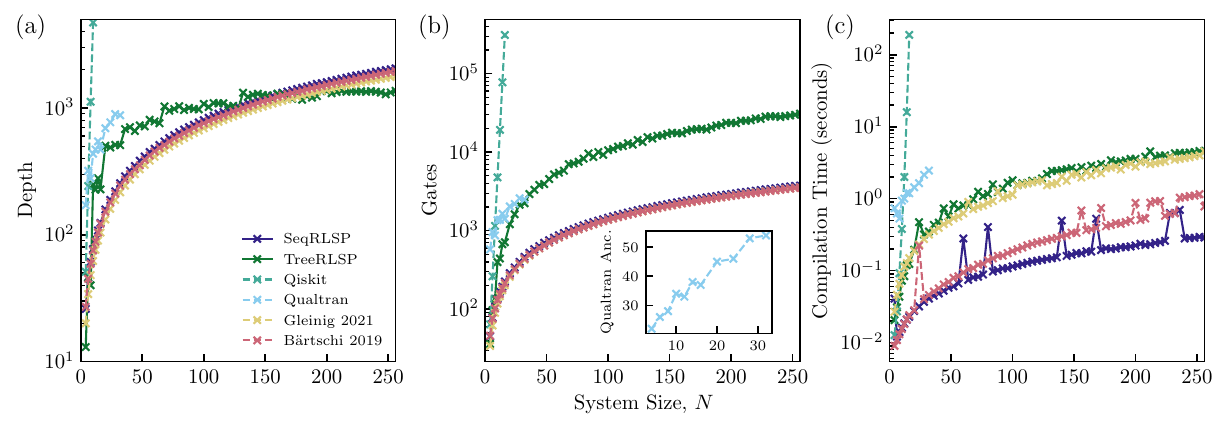}
    \caption{Preparation cost of W-state for various system sizes $N$. We compare \SeqRLSP, \TreeRLSP, \emph{Qiskit}, \emph{Qualtran}, \citet{gleinig2021efficient}, and \citet{bartschi2019deterministic}. (a) Quantum circuit depth. (b) Total gate count. The inset show the number of ancillae used by \emph{Qualtran}. (c) Compilation time in seconds.}
    \label{fig:Wstate}
\end{figure}

\subsection{Complement States}
As highlighted earlier, complement RLSs are particularly well suited for our framework as they can be expressed efficiently as RLSs. In particular, if a state $\ket{\Psi}$ can be represented as an MPS with bond dimension $\chi$, its complement $\ket{\Psi_c} = \ket{+}^{\otimes N} - \ket{\Psi}$ can be expressed as an MPS with bond dimension at most $\chi + 1$~(see \cref{thm:complementRLS}). To the best of our knowledge, there are no specialized compilers designed to efficiently generate complements of otherwise easily preparable states.
In \cref{fig:Dicke2comp}, we report numerical results for the preparation of the complement of a Dicke state with $k=2$.
We compare \SeqRLSP and \TreeRLSP against \emph{Qiskit}~\cite{qiskit}, \emph{Qualtran}~\cite{qualtran, Plesch2011}, and \citet{gleinig2021efficient}.
Since the complement of a Dicke$-2$ state is not sparse, sparse-state preparation methods such as \emph{Qualtran} and \citet{gleinig2021efficient} will exhibit exponential scaling with system size $N$ both in quantum resources and compile time, and are, thus, limited to small $N$.
In contrast, \SeqRLSP and \TreeRLSP remain efficient and are the only approaches capable of handling $N > 20$ for this task.
\begin{figure}[t]
    \centering
    \includegraphics[width=\linewidth]{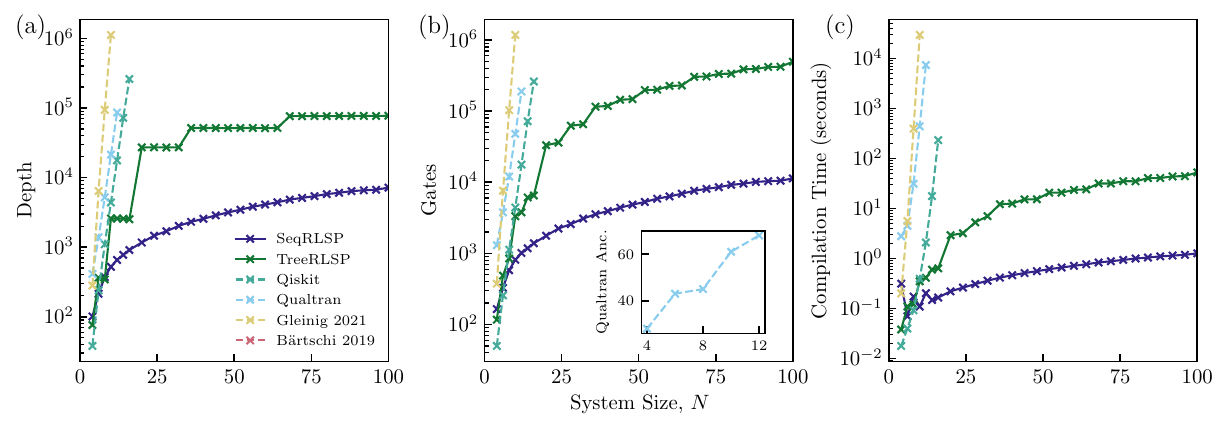}
    \caption{Preparation cost of Dicke$-2$ state complement for various system sizes $N$. We compare \SeqRLSP, \TreeRLSP, \emph{Qiskit}, \emph{Qualtran}, \citet{gleinig2021efficient}, and \citet{bartschi2019deterministic}. (a) Quantum circuit depth. (b) Total gate count. The inset show the number of ancillae used by \emph{Qualtran}. (c) Compilation time in seconds.}
    \label{fig:Dicke2comp}
\end{figure}

\subsection{Random Uniform Superpositions}
Preparing a random uniform superposition of strings is precisely the task for which sparse state preparation algorithms~\cite{qualtran, gleinig2021efficient} are designed.
In contrast, our approaches \SeqRLSP and \TreeRLSP are optimized to exploit underlying structure for compression, an advantage that is absent for random states.
Thus, this subsection examines the \emph{worst-case} performance of our algorithms.
Specifically, representing a superposition of $\spar$ random bitstrings requires an MPS with bond dimension $\chi = \spar$ in the absence of structure.

In \cref{fig:rand_str}, we report the cost of preparing random superpositions for system sizes $N = 10$ and $N = 20$.
As expected, the maximum MPS bond dimension generally grows as $\chi = \spar$. For $N = 10$, a plateau occurs after $\spar \approx 32$. This behavior arises because, for an MPS with $N$ sites, the maximum bond dimension is bounded by $\chi = \min(2^{\lfloor N / 2 \rfloor}, \spar)$ due to the \emph{staircase} structure of the MPS. Stepwise increases in circuit depth and gate count appear whenever $\spar$ crosses the next power of two, particularly evident for $N = 20$.

Notably, for $N = 10$ and $\spar \approx 1024$, approaching a uniform superposition of all bitstrings (i.e., a product state), \SeqRLSP identifies the underlying (complement) structure, producing a sudden decrease in depth and gate count.
The inset in \cref{fig:rand_str} (a) further illustrates that at very large $\spar$, the \emph{staircase} structure reappears, indicating that our approach effectively detects when the resulting state can be compressed, as the complement of $(N-\spar)$ random strings.  In contrast, sparse-state methods are not able to exploit this structure: \emph{Qualtran} fails to recognize the $(N=10, \spar=1024)$ product state, and \citet{gleinig2021efficient} does not converge within a running time exceeding $10^5$ seconds.

\begin{figure}[!ht]
    \centering
    \includegraphics[width=\linewidth]{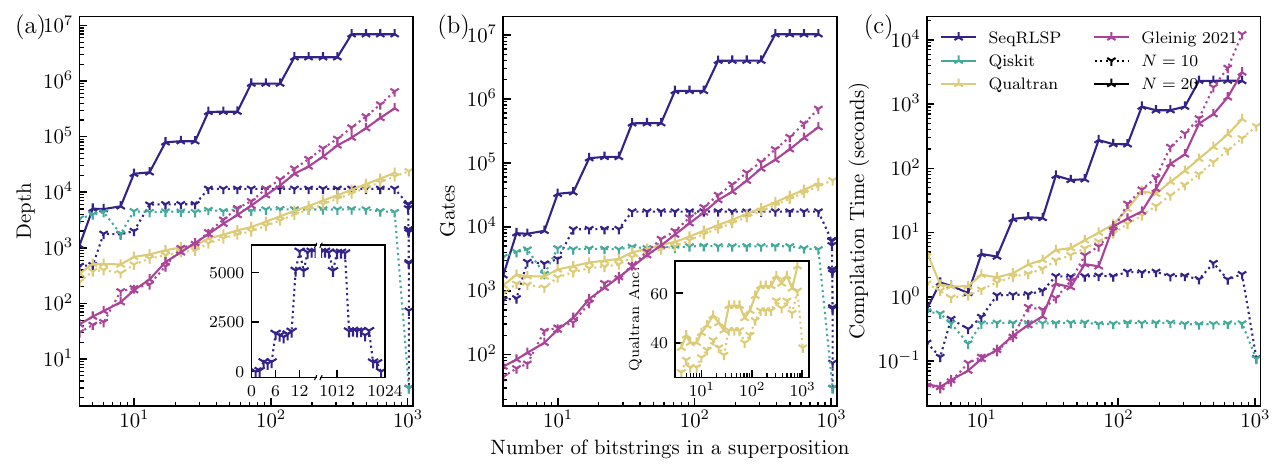}
    \caption{Preparation cost of a quantum states of $N=10$ and $N=20$ qubits consisting of a uniform superposition of $\spar$ random bitstrings. We compare \SeqRLSP, \emph{Qiskit}, \emph{Qualtran}, and \citet{gleinig2021efficient}. (a) Quantum circuit depth. The inset shows the \SeqRLSP behavior at small and large $\spar$. (b) Total gate count. The inset show the number of ancillae used by \emph{Qualtran}. (c) Compilation time in seconds.}
    \label{fig:rand_str}
\end{figure}

\section{Conclusion}
\label{sec:conclusions}
We developed a quantum circuit compiler pipeline that sits between universal but structure-oblivious state preparation routines and bespoke, family-specific constructions.
We targeted RLS and their complements—uniform superpositions of strings that a given regular language accepts, or rejects for the complements.
The pipeline combines a regular-language-based frontend with a DFA/MPS intermediate representation, allowing users to specify the target state via regular expression, DFA, or finite set of strings, without having to identify a named state family in advance.

From this IR, we outline two hardware-aware backends, \SeqRLSP and \TreeRLSP. \SeqRLSP targets linear nearest-neighbor hardware and achieves linear depth in the system size, while \TreeRLSP exploits all-to-all connectivity to reduce the depth to logarithmic in the system size.
In both cases, the DFA/MPS IR is central: by minimizing the DFA and turning the MPS into an isometries network, the compiler can detect and exploit compressibility when it exists.
The DFA IR provides a fast and efficient path from user-level descriptions to MPS constructions, bridging the gap to MPS synthesis strategies that usually assume a good MPS is already available.

Our numerical experiments validate the theoretical resource bounds and show that the compiler can approach the performance of tailored methods for preparing Dicke states, while also recognizing and exploiting structure from unstructured descriptions in settings where generic sparse state preparation routines do not.  
Moreover, they show the power of the complement specification and efficient preparation, which is an entirely novel feature of our approach.

A key design benefit of our approach is its modularity. 
The DFA/MPS IR cleanly separates the user-facing specification from the backend realization, so improvements in regex-to-DFA tools, DFA minimization, MPS synthesis, or isometry decomposition can be integrated without changing the frontend. 
Future work includes extending our methods beyond uniform RLS to states with arbitrary complex-valued amplitudes, and exploring more advanced MPS-to-circuit mappings—such as measurement-and-feedback schemes or more direct syntheses of $1$- and $2$-qubit gate sequences—to further tighten the scaling with $\chi$ and match specialized constructions even more closely.

\acknowledgments{This work was partly carried out during the MPQ Theory Hackathon. The authors thank Khurshed P. Fitter for initial discussions. A.B. thanks Davide Conficconi for valuable feedback and helpful discussions.
The work at MPQ is supported from the German Federal Ministry of Education and Research (BMBF) through the funded project
ALMANAQC, grant number 13N17236 within the research program “Quantum Systems”. The research is partly funded by THEQUCO as part of the Munich Quantum Valley, which is supported by the Bavarian state government with funds from the Hightech Agenda Bayern Plus.}
\medskip

\bibliographystyle{unsrtnat}
\bibliography{bibliography}

\end{document}